%% file: main.tex
\def\year{2021}\relax
\documentclass[letterpaper]{article} 
\usepackage{aaai21}  
\usepackage{times}  
\usepackage{helvet} 
\usepackage{courier}  
\usepackage[hyphens]{url}  
\usepackage{graphicx} 
\usepackage{natbib}  
\usepackage{caption} 
\usepackage{amssymb}

\usepackage{amsmath,amsfonts,amsthm,bm} 

\DeclareMathOperator*{\argmin}{arg\,min}
\usepackage{color}
\usepackage{lipsum} 

\DeclareMathOperator*{\E}{\mathbb{E}}

\usepackage{enumerate}
\usepackage{enumitem}

\makeatletter
\newcommand{\RN}[1]{\textup{\uppercase\expandafter{\romannumeral#1}}}

\newcommand{\Rmnum}[1]{\expandafter\@slowromancap\romannumeral #1@}
\makeatother

\newtheorem{theorem}{Theorem}

\newtheorem{lemma}{Lemma}
\newtheorem{proposition}{Proposition}
\newtheorem{corollary}{Corollary}

\theoremstyle{definition}
\newtheorem{definition}{Definition}
\newtheorem{assumption}{Assumption}
\newtheorem{remark}{Remark}

\usepackage{graphicx}

\usepackage{tikz}
\usepackage{subfigure}

\usepackage{mathtools}

\newcommand{\R}{\mathbb R}

\newcommand{\M}{\mathcal M}
\newcommand{\K}{\mathcal K}
\newcommand{\AK}{A_{\Kb}}
\newcommand{\Kb}{\mathbb K}

\newcommand{\A}{\mathcal A}

\newcommand{\X}{\mathcal X}
\newcommand{\U}{\mathcal U}
\newcommand{\W}{\mathcal W}

\usepackage{accents}
\newcommand{\dbtilde}[1]{\accentset{\approx}{#1}}

\usepackage{dsfont}
\newcommand{\one}{\mathds{1}}
\newcommand{\nbf}{\noindent\textbf}

\usepackage{algorithm}
\usepackage[algo2e,ruled,noend]{algorithm2e}

\title{Online Optimal Control with Affine Constraints}
\author {
        Yingying Li,\textsuperscript{\rm 1}
        Subhro Das,\textsuperscript{\rm 2} 
        Na Li\textsuperscript{\rm 1} \\
}
\affiliations {
    \textsuperscript{\rm 1} John A. Paulson School of Engineering and Applied Sciences, Harvard University \\
    \textsuperscript{\rm 2} MIT-IBM Watson AI Lab, IBM Research \\
    yingyingli@g.harvard.edu, 
    subhro.das@ibm.com, nali@seas.harvard.edu
}
\begin{document}

\maketitle

\begin{abstract}
This paper considers online optimal control with affine constraints on the states and  actions under linear dynamics with bounded random disturbances. The system dynamics and constraints are assumed to be known and time invariant but the convex stage cost functions change adversarially.
To solve this problem,  we propose Online Gradient Descent with Buffer Zones (OGD-BZ). Theoretically, we show that OGD-BZ with proper parameters can guarantee the system to satisfy all the constraints despite any admissible disturbances. Further, we investigate the policy regret of OGD-BZ, which compares OGD-BZ's performance with the performance of the optimal linear policy in hindsight. We show that OGD-BZ can achieve a policy regret upper bound that is square root of the horizon length multiplied by some logarithmic terms of the horizon length under proper algorithm parameters.
\end{abstract}


\section{Introduction}

Recently, there is a lot of interest  in solving control problems by  learning-based  techniques, e.g. online learning and reinforcement learning \cite{agarwal2019online,li2019online,ibrahimi2012efficient,dean2018regret,fazel2018global,yang2019provably,li2019distributed}. This is motivated by    
 applications such as data centers \cite{lazic2018data,li2019onlineMDP}, robotics \cite{fisac2018general}, autonomous vehicles \cite{sallab2017deep}, power systems \cite{chen2021reinforcement}, etc. For real-world implementation, it is crucial to design safe algorithms that ensure the system to satisfy certain (physical) constraints
despite unknown disturbances. For example,  temperatures in  data centers should be within certain ranges to reduce task failures despite  disturbances from unmodeled heat sources,    quadrotors should avoid collisions even when perturbed by wind, etc.  In addition to safety, many applications involve time-varying environments, e.g. varying electricity prices, moving targets, etc.  Hence,  safe algorithms  should not be over-conservative and should adapt to varying environments for desirable    performance.

In this paper, we  design safe algorithms for time-varying environments by considering the following constrained online optimal control problem. Specifically, we consider a linear system with random disturbances,
\begin{align}\label{equ: system}
    x_{t+1}=A x_t + B u_t +w_t, \quad t\geq 0,
\end{align} 
where disturbance $w_t$ is random and satisfies $\|w_t\|_\infty \leq \bar w$. Consider affine constraints on the state $x_t$ and the action $u_t$:
\begin{align}\label{equ: sec1 affine constraints}
    D_x x_t \leq d_x, \quad D_u u_t \leq d_u, \quad \forall\  t\geq 0.
\end{align}
For simplicity, we assume the system parameters $A, B, \bar w$ and the constraints are known.
At  stage $0\leq t\leq T$, a convex cost function $c_t(x_t,u_t)$ is adversarially generated and the decision maker selects a feasible action $u_t$ before $c_t(x_t, u_t)$ is revealed.   We aim to achieve  two goals simultaneously:  (i) to minimize the sum of the adversarially varying costs, (ii) to satisfy the constraints \eqref{equ: sec1 affine constraints} for all $t$ despite the disturbances.

There are many studies to address each goal separately but  lack results on both goals together as discussed below.

Firstly, there is recent progress on online optimal control to address Goal (i).  A commonly adopted performance metric is policy regret, which compares the online cost with the cost of the optimal linear policy in hindsight \cite{agarwal2019online}. Sublinear policy regrets have been achieved for linear systems  with either stochastic disturbances \cite{cohen2018online,agarwal2019logarithmic} or adversarial disturbances \cite{agarwal2019online,foster2020logarithmic,goel2020power,goel2020regret}. However, most literature only considers the unconstrained control problem. Recently, \citet{nonhoff2020online} studies constrained online optimal control but  assumes no disturbances.

Secondly, there are many papers from the control community to address Goal (ii): constraints satisfaction. Perhaps the most famous algorithms  are Model Predictive Control (MPC)  \cite{rawlings2009model} and its variants, such as robust MPC which guarantees (hard) constraint satisfaction in the presence of disturbances \cite{bemporad1999robust,kouvaritakis2000efficient,mayne2005robust,limon2010robust,zafiriou1990robust} as well as stochastic MPC which considers soft constraints and allows constraints violation \cite{oldewurtel2008tractable,mesbah2016stochastic}. However, there lack algorithms with both regret/optimality guarantees and constraint satisfaction guarantees. 

Therefore, an important question remains to be addressed:
\textit{Q: how to design online algorithms to both satisfy the constraints despite  disturbances and  yield $o(T)$ policy regrets?}

\subsubsection{Our Contributions} In this paper, we answer the question above by proposing an online control algorithm: Online Gradient Descent with Buffer Zones (OGD-BZ). To develop OGD-BZ, we first convert the constrained online optimal control problem into an online convex optimization (OCO) problem with temporal-coupled stage costs and temporal-coupled stage constraints, and then convert the temporal-coupled OCO problem into a classical OCO problem. The problem conversion  leverages the techniques from recent unconstrained online control literature and robust  optimization literature. Since the  conversion is not exact/equivalent, we tighten the constraint set by adding buffer zones to account for approximation errors caused by the problem conversion. We then apply classical OCO method OGD to solve the problem and call the resulting algorithm as OGD-BZ.

Theoretically, we show that, with proper parameters, OGD-BZ can ensure all the states and actions to satisfy the constraints \eqref{equ: sec1 affine constraints} for any  disturbances bounded by $\bar w$.  
In addition, we  show that OGD-BZ's policy regret can be bounded by $\tilde O(\sqrt T)$ for general convex cost functions $c_t(x_t,u_t)$ under proper assumptions and parameters.  As far as we know, OGD-BZ is the first algorithm with theoretical guarantees on both  sublinear policy regret and robust constraint
satisfaction.  Further, our theoretical results explicitly characterize a trade-off between the constraint satisfaction and the low regret when deciding  the size of the buffer zone of OGD-BZ. That is, a larger buffer zone, which indicates a more conservative search space, is preferred for constraints satisfaction; while a smaller buffer zone is preferred for low regret.

\subsubsection{Related Work} We provide more literature review below.


\noindent\textit{Safe reinforcement learning.} There is a rich body of literature on safe RL and safe learning-based control that studies how to learn optimal policies without violating constraints and without knowing the system \cite{fisac2018general,aswani2013provably,wabersich2018linear,garcia2015comprehensive,cheng2019end,zanon2019safe,fulton2018safe}. Perhaps the most relevant paper is \citet{dean2019safely}, which  proposes algorithms to learn optimal linear policies for a constrained linear quadratic regulator problem. However, most theoretical guarantees  in the safe RL literature require  time-invariant environment and there lacks policy regret analysis when facing time-varying objectives. This paper addresses the time-varying objectives but considers known system dynamics. It is our ongoing work to combine both safe RL and our approach to design safe learning algorithms with policy regret guarantees in time-varying problems. 

Another important notion of safety is  the system stability, which  is also  studied in the safe RL/learning-based control literature \cite{dean2018regret,dean2019sample,chow2018lyapunov}.


\noindent\textit{Online convex optimization (OCO).} \citet{hazan2019introduction} provides a review on classical (decoupled) OCO. OCO with memory considers  coupled costs and decoupled constraints  \cite{anava2015online,li2020onlineTAC,li2020leveraging}. The  papers on OCO with coupled constraints usually allow constraint violation \cite{yuan2018online,cao2018virtual, kveton2008online}. Besides, OCO  does not consider system dynamics or  disturbances. 


\noindent\textit{Constrained optimal control.}  Constrained optimal control enjoys a long history of research.  Without disturbances, it is known that the optimal controller for linearly constrained linear quadratic regulator  is  piecewise linear \cite{bemporad2002explicit}. With disturbances (as considered in this paper), the problem is much more challenging.  Current methods such as robust  MPC \cite{limon2008design,limon2010robust,rawlings2009model} and stochastic MPC \cite{mesbah2016stochastic,oldewurtel2008tractable} usually deploy linear policies for fast computation even though linear policies are suboptimal. Besides, most theoretical analysis of robust/stochastic MPC focus on stability, recursive feasibility, and constraints satisfaction, instead of policy regrets. 

\subsubsection{Notations and Conventions}
We let $\|\cdot \|_1, \|\cdot \|_2, \|\cdot\|_\infty$ denote the $L_1, L_2, L_\infty$ norms respectively for  vectors and matrices.  Let $\one_n$ denote an all-one vector in $\R^n$. For two vectors $a, b \in \R^n$, we write $a\leq b$ if $a_i \leq b_i$ for any entry $i$. Let $\text{vec}(A)$ denote the vectorization of matrix $A$. For better exposition, some  bounds use $\Theta(\cdot)$ to omit  constants that do not depend on  $T$ or  the  problem dimensions explicitly.


\section{Problem Formulation}
In this paper, we consider an online optimal control problem with linear dynamics and affine constraints. 
Specifically,  at each stage $t\in\{0,1, \dots, T\}$, an agent observes the current state $x_t$ and implements an action $u_t$, which incurs a cost $c_t(x_t, u_t)$. The stage cost function $c_t(\cdot, \cdot)$ is generated adversarially and revealed to the agent after the action $u_t$ is taken. The system evolves to the next state according to \eqref{equ: system}, 
where $x_0$ is fixed,  $w_t$ is a random disturbance bounded by $w_t \in \mathcal W=\{w\in \R^n:  \|w\|_\infty \leq \bar w\}$, and states and actions should satisfy the affine constraints \eqref{equ: sec1 affine constraints}. We denote the corresponding constraint sets as
$$
\mathcal X\!=\!\{x \in\R^n\!\!:\!  D_x x \leq \!d_x\}, \quad \mathcal U\!=\!\{u\in \R^m\!\!:\! D_u u \leq\! d_u\},
$$ 
where $d_x \in \R^{k_x}$ and $d_u \in \R^{k_u}$. Define $k_c=k_x+k_u$ as the total number of the constraints.

For simplicity, we consider that the parameters $A, B, \bar w, D_x, d_x, D_u, d_u$ are known a priori and that the initial value satisfies $x_0=0$. We leave the study of unknown parameters and general $x_0$  for the future.

\begin{definition}[Safe controller]\label{def: feasible controller}
	Consider a controller (or an algorithm) $\A$ that chooses action $u_t^{\A}\in \U$ based on history states $\{x_k^{\A}\}_{k=0}^t$  and cost functions $\{c_k(\cdot, \cdot)\}_{k=0}^{t-1}$. The controller $\A$ is called \textit{safe} if $x_t^{\A}\in\mathcal X$ and $u_t^{\A}\in \mathcal U$ for all $0\leq t \leq T$ and all disturbances $\{w_k\in \mathcal W\}_{k=0}^T$.
\end{definition}
\noindent Define the total cost of a safe algorithm/controller $\A$ as:
\begin{align}
J_T(\A)=\E_{\{w_k\}}\left[\sum_{t=0}^T c_t(x_t^{\A}, u_t^{\A})\right].
\end{align}

\subsubsection{Benchmark Policy and Policy Regret} In this paper,  we consider linear policies of the form $u_t=-K x_t$ as our benchmark policy for simplicity, though the optimal policy for the constrained  control of noisy systems may be  nonlinear \cite{rawlings2009model}. We leave the discussion on nonlinear policies as future work. 

Based on \cite{cohen2018online}, we define strong stability, which is a quantitative version of   stability  and is commonly introduced to ease non-asymptotic regret analysis in the online control literature \cite{agarwal2019logarithmic,agarwal2019online}.
\begin{definition}[Strong Stability]\label{def: diagonally strongly stable}
A linear controller $u_t=-K x_t$  is $(\kappa, \gamma)$-strongly stable for $\kappa\geq 1$ and $\gamma \in (0,1]$ if there exists a  matrix $L$ and an invertible matrix $Q$ such that $A-BK=Q^{-1}L Q$, with $\|L \|_2\leq 1-\gamma$ and $\max(\|Q\|_2, \|Q^{-1}\|_2, \|K\|_2)\leq \kappa$.
\end{definition}
As shown in \citet{cohen2018online},  strongly stable controllers can be computed efficiently by SDP formulation.

Our benchmark policy class includes any linear controller $u_t=-K x_t$ satisfying the conditions below:
\begin{align*}
\mathcal K=\{ K: \text{$K$  is safe and $(\kappa, \gamma)$-strongly stable}\},
\end{align*}
where $K$ is called safe if the controller $u_t=-K x_t$ is safe according to Definition \ref{def: feasible controller}.

The policy regret of online algorithm $\A$ is defined as:
\begin{equation}
\text{Reg}(\A) =J_T(\A)-\min_{K \in \K} J_T(K).
\end{equation}

\subsubsection{Assumptions and Definitions} For the rest of the paper, we define $\kappa_B = \max(\|B\|_2, 1)$.   In addition, we introduce the following assumptions on the disturbances and the cost functions, which are standard  in  literature \cite{agarwal2019logarithmic}.
\begin{assumption}\label{ass: wt}
	$\{w_t\}$ are i.i.d. and bounded by $\|w_t \|_\infty \leq \bar w$, where $\bar w>0$.\footnote{The results of this paper can be extended to adversarial noises.} 
\end{assumption}
\begin{assumption}\label{ass: bounded Hessian largest evalue}
	For any $t\geq 0$, cost function $c_t(x_t, u_t)$ is convex and differentiable with respect to $x_t$ and $u_t$. Further, there exists $G>0$, such that for any $\|x\|_2\leq b$, $\|u\|_2\leq b$, we have $\|\nabla_x c_t(x,u)\|_2\leq Gb$ and $\|\nabla_u c_t(x,u)\|_2\leq Gb$.
\end{assumption}

Next, we define strictly and loosely safe controllers.

\begin{definition}[Strict and loose safety]\label{def: strictly feasible policy}
	A safe controller $\mathcal A$ is called  \textit{$\epsilon$-strictly safe} for some $\epsilon>0$ if $D_x x_t^{\A}\leq d_x-\epsilon \one_{k_x}$ and $D_u u_t^{\A}\leq d_u-\epsilon \one_{k_u}$ for all $0\leq t \leq T$ under any disturbance sequence $\{w_k\in \mathcal W\}_{k=0}^T$.
	
	A  controller $\mathcal A$ is called  \textit{$\epsilon$-loosely safe}  for some $\epsilon>0$ if $D_x x_t^{\A}\leq d_x+\epsilon \one_{k_x}$ and $D_u u_t^{\A}\leq d_u+\epsilon \one_{k_u}$ for all $0\leq t \leq T$ under any disturbance sequence $\{w_k\in \mathcal W\}_{k=0}^T$.
	
\end{definition}

In the following, we assume the existence of a \textit{strictly} safe linear policy. The existence of a safe linear policy is necessary since otherwise our  policy regret  is not well-defined. The existence of a strictly safe policy provides some flexibility for the approximation steps in our algorithm design and is a common assumption in constrained optimization and control \cite{boyd2004convex,limon2010robust}. 
\begin{assumption}\label{ass: K strictly feasible}
	There exists $K_*\in \mathcal K$ such that the policy $u_t=-K_* x_t$ is $\epsilon_*$-strictly safe for some $\epsilon_*>0$.
\end{assumption}
Intuitively, Assumption \ref{ass: K strictly feasible} requires  the sets $\X$ and $\U$ to  have non-empty interiors and that the disturbance set $\W$ is  small enough so that a disturbed linear system $x_{t+1}=(A-BK_*)x_t +w_t$ stays in the interiors of $\X$ and $\U$ for any $\{w_k\in\W\}_{k=0}^T$. In addition, Assumption \ref{ass: K strictly feasible}  implicitly assumes that 0 belongs to  the interiors of $\X$ and $\U$ since we let $x_0=0$. Finally, 
though it is challenging to verify Assumption \ref{ass: K strictly feasible}  directly, there are numerical verification methods, e.g. by solving  linear matrix inequalities (LMI) programs \cite{limon2010robust}.\footnote{\cite{limon2010robust} provides an LMI program to compute a near-optimal linear controller for a time-invariant constrained control problem, which can be used to verify the existence of a safe solution. To verify Assumption 3, one could  run the LMI program   with the constraints tightened by $\epsilon$ and continue to reduce $\epsilon$ if no solution is found until $\epsilon$ is smaller than a certain threshold.}


\section{Preliminaries}\label{sec: prelim}
This section briefly reviews the unconstrained online optimal control and robust  constrained optimization literature,  techniques from which motivate our algorithm design.

\subsection{Unconstrained Online Optimal Control}
In our setting, if one considers $\X=\R^n$ and $\U=\R^m$, then the problem reduces to an unconstrained online optimal control. For such unconstrained online control problems, \citet{agarwal2019logarithmic,agarwal2019online} propose a disturbance-action policy class to design an online policy.

\begin{definition}[Disturbance-Action Policy \cite{agarwal2019logarithmic}]\label{def: disturbance action}
Fix an arbitrary $(\kappa,\gamma)$-strongly stable matrix $\Kb$ a priori. 
	Given an $H\in\left\{1,2,\ldots,T\right\}$,  a disturbance-action policy  defines the control policy as:
	\begin{equation} 
	u_t=-\Kb x_t + \sum_{i=1}^H M^{[i]}w_{t-i}, \quad \ \forall\, t\geq 0, \label{eq:Dis_act_policy}
	\end{equation}	where, $M^{[i]}\in\R^{m \times n}$ and $w_t=0$ for $t\leq 0$. Let $\bm M=\{M^{[i]}\}_{i=1}^H$ denote the list of parameter matrices for the disturbance-action policy.\footnote{The disturbance-action policy is mainly useful for non-zero disturbances. Nevertheless, our theoretical results do not require $w_t\not =0$ because for no-disturbance systems, any strongly stable controller $u_t=-\mathbb K x_t$ will only result in a constant $O(1)$ regret.}
\end{definition}
For the rest of the paper, we will fix $\mathbb K$ and discuss how to choose parameter $\bm M$. 
 In \citet{agarwal2019logarithmic},  a bounded convex constraint set  on policy  $\bm M$ is introduced for technical simplicity and without loss of generality:\footnote{This is without loss of generality because \cite{agarwal2019online} shows that any $(\kappa, \gamma)$-strongly stable linear policy can be approximated by a disturbance-action policy in  $\M_2$.}
\begin{align}\label{equ: M_2 def}
    \M_2 \!=\!\{ \bm M\!=\!\{M^{[i]}\}_{i=1}^H\!:  \! \|M^{[i]}\|_2\! \leq \!\kappa^3\kappa_B (1\!-\!\gamma)^i,  \forall i\}
\end{align}

The next proposition introduces state and action  approximations  when implementing   disturbance-action policies.
\begin{proposition}[\cite{agarwal2019online}]
    \label{prop: approx}
	When implementing a disturbance-action policy \eqref{eq:Dis_act_policy}  with time-varying $\bm M_t=\{M_t^{[i]}\}_{i=1}^H$ at each stage $t\geq 0$, the states and actions satisfy:
	\begin{align}\label{equ: approx state and action}
	&	x_t = A_{\Kb}^H x_{t-H}+ \tilde x_t \ \ \text{and } \  	u_t = -\Kb A_{\Kb}^H x_{t-H}+ \tilde u_t, 
	\end{align}
	where  $\AK=A-B\Kb$. The approximate/surrogate state and  action, $\tilde x_t$ and $\tilde u_t$, are defined as:
	\begin{align*}
	& \tilde x_t=\sum_{k=1}^{2H} \Phi_{k}^x(\bm M_{t-H:t-1})w_{t-k},\\
	&	 \tilde u_t=-\Kb \tilde x_t + \sum_{i=1}^H M_t^{[i]}w_{t-i}=\sum_{k=1}^{2H} \Phi_{k}^u(\bm M_{t-H:t})w_{t-k},\\
	&\Phi_{k}^x(\bm M_{t-H:t-1})\!=\!	\AK^{k\!-\!1}\one_{(k\!\leq\! H)}\!+\! \sum_{i=1}^H \!\AK^{i\!-\!1}B M_{t\!-\!i}^{[k\!-\!i]}\one_{(1\!\leq \!k\!-\!i\!\leq\! H)} \\
	& 	\Phi_{k}^u(\bm M_{t-H:t})=	M_t^{[k]}\one_{(k\leq H)}-\Kb\Phi_{k}^x(\bm M_{t-H:t-1}),
	\end{align*}
	where $\bm M_{t-H:t}\coloneqq \{\bm M_{t-H}, \dots, \bm M_t\}$,  the   superscript $k$ in $\AK^k$ denotes the $k$th power of  $\AK$,  and  $M_t^{[k]}$ with superscript $[k]$ denotes the $k$th matrix in  list  $\bm M_t$.  Further, define  $\mathring \Phi_k^x(\bm M)= \Phi_{k}^x(\bm M, \dots, \bm M),\,\mathring \Phi_k^u(\bm M)= \Phi_{k}^u(\bm M, \dots, \bm M)$. 

\end{proposition}

Notice that $\tilde x_t$ and $\tilde u_t$ are affine functions of $\bm M_{t-H:t}$. Based on $\tilde x_t$ and $\tilde u_t$, \citet{agarwal2019logarithmic} introduces an approximate  cost function:
$$ f_t(\bm M_{t-H:t})= \E[c_t(\tilde x_t, \tilde u_t)],$$
which is convex with respect to $\bm M_{t-H:t}$ since $\tilde x_t$ and $\tilde u_t$ are affine functions of $\bm M_{t-H:t}$ and $c_t(\cdot, \cdot)$ is convex. 


\begin{remark}
The disturbance-action policy is   related to  \textit{affine disturbance feedback } in  stochastic MPC \cite{oldewurtel2008tractable,mesbah2016stochastic}, which also considers policies that are linear with disturbances  to convexify the control problem in  MPC's lookahead horizon.

\end{remark}

\nbf{OCO with memory.} 
  In \citet{agarwal2019logarithmic}, the unconstrained online optimal control problem is converted into  \textit{OCO with memory}, i.e. at each stage $t$, the agent selects a policy $\bm M_t\in \M_2$ and then incurs a cost $ f_t(\bm M_{t-H:t})$. Notice that the cost function at stage $t$ couples the current policy $\bm M_t$ with the $H$-stage historical policies $\bm M_{t-H:t-1}$, but the constraint set  $\M_2$ is decoupled and only depends on the current $\bm M_t$.
  
  To solve this ``OCO with memory'' problem, \citet{agarwal2019logarithmic} defines  decoupled  cost functions
  \begin{equation}\label{equ: mathring f_t(M_t)}
\mathring f_t(\bm M_t)\coloneqq f_t(\bm M_t, \dots, \bm M_t),
\end{equation}
by letting the $H$-stage historical policies be identical to the current policy. Notice that $\mathring f_t(\bm M_t)$ is still convex.  Accordingly, the OCO with memory  is  reformulated as a classical OCO problem with stage cost  $\mathring f_t(\bm M_t)$, which is solved by classical OCO algorithms such as online gradient descent (OGD) in \citet{agarwal2019logarithmic}. The stepsizes of OGD are chosen to be sufficiently small so that the variation between the current policy $\bm M_t$  and the $H$-stage historical policies $\bm M_{t-H}, \dots, \bm M_{t-1}$ is sufficiently small, which guarantees a small approximation error between $\mathring f_t(\bm M_t)$ and $f_t(\bm M_{t-H:t})$, and thus low regrets. For more details, we refer the reader to \citet{agarwal2019logarithmic}.

\subsection{Robust  Optimization with Constraints}\label{subsec: robust opt}
Consider a robust optimization problem with linear constraints \cite{ben2009robust}:
\begin{align}\label{equ: robust opt}
\min_{x} &\ f(x)\quad 
\text{s.t. }  \ a_i^\top x \leq b_i, \ \forall\, a_i \in \mathcal C_i, \ \forall\, 1\leq i \leq k,
\end{align}
where the (box) uncertainty sets are defined as $\mathcal C_i= \{a_i\!=\!\tilde a_i \!+\! P_i z\!:\! \|z\|_\infty\! \leq\! \bar z\}$ for any $i$. Notice that  
the robust constraint $
\{a_i^\top x\leq b_i, \ \forall \, a_i \in \mathcal C_ i\}$ is equivalent to the  standard constraint $ \{ \sup_{a_i \in \mathcal C_i} [a_i^\top x] \leq b_i\}$.
Further, one can derive
\begin{align}
&\sup_{a_i \in \mathcal C_i} a_i^\top x=\sup_{\|z\|_\infty \leq \bar z} (\tilde a_i+P_i z)^\top x\nonumber\\
=\ & \tilde a_i^\top x+ \sup_{\|z\|_\infty \leq \bar z} z^\top (P_i^\top x)  = \tilde a_i^\top x+ \|P_i^\top x\|_1 \bar z\label{equ: robust opt L1}
\end{align}
Therefore, the robust optimization \eqref{equ: robust opt} can be equivalently reformulated as the linearly constrained optimization below:
\begin{align*}
\min_{x} &\ f(x) \quad \text{s.t. } \ \tilde a_i^\top x+ \|P_i^\top x\|_1 \bar z\leq b_i,  \ \forall\, 1\leq i \leq k.
\end{align*}


\section{Online Algorithm Design}\label{sec: algorithm}

This section introduces our  online algorithm design for   online disturbance-action policies (Definition \ref{def: disturbance action}).

Roughly speaking, to develop our online algorithm, we first convert the constrained online optimal control   into \textit{OCO with memory and coupled constraints}, which is later converted into classical OCO and solved by OCO algorithms.  The conversions leverage the approximation and the  reformulation techniques in the \textbf{Preliminaries}.  During the conversions, we ensure that the outputs of the OCO algorithms are safe for the original control problem. This is achieved by tightening the original constraints (adding buffer zones) to allow for approximation errors. Besides, our method ensures small approximation errors and thus  small   buffer zones  so that the optimality/regret is not sacrificed significantly for  safety. The details of algorithm design are discussed below.

\subsubsection{Step 1: Constraints on  Approximate States and Actions} When applying the disturbance-action policies \eqref{eq:Dis_act_policy},  we can use \eqref{equ: approx state and action} to rewrite the state constraint  $x_{t+1}\in\X$ as
\begin{align}\label{equ: state constraint with approximate state decoupling}
& D_x \AK^H x_{t-H+1}+ D_x \tilde x_{t+1}\leq d_x, \ \  \forall\, \{w_k \in \W\}_{k=0}^T,
\end{align}
where $\tilde x_{t+1}$ is the approximate state. Note that the term $D_x \AK^H x_{t-H+1}$ 
decays exponentially with $H$. If there exists $H$ such that $D_x \AK^H x_{t-H+1}\leq \epsilon_1 \one_{k_x}$, $\forall\,\{w_k \in \W\}_{k=0}^T$, then a tightened constraint on the approximate state, i.e.
\begin{equation}\label{equ: constraint on tilde x_t}
D_x \tilde x_{t+1}\leq d_x-\epsilon_1\one_{k_x}, \quad \forall\, \{w_k \in \W\}_{k=0}^T,
\end{equation}
can guarantee the original constraint on the true state \eqref{equ: state constraint with approximate state decoupling}.

The action constraint  $u_t\in\U$ can similarly be converted into a tightened constraint on the approximate action $\tilde u_t$, i.e.
\begin{align}\label{equ: constraint on tilde u_t}
D_u \tilde u_t\leq d_u-\epsilon_1\one_{k_u}, \quad \forall\, \{w_k \in \W\}_{k=0}^T,
\end{align} 
if  $D_u (-\Kb\AK^H x_{t-H})\leq \epsilon_1\one_{k_u}$ for any  $\{w_k \in \W\}_{k=0}^T$.

\subsubsection{Step 2: Constraints on the Policy Parameters} 
Next, we reformulate the robust constraints \eqref{equ: constraint on tilde x_t} and \eqref{equ: constraint on tilde u_t} on $\tilde x_{t+1}$ and $\tilde u_t$  as polytopic constraints on policy parameters $\bm M_{t-H:t}$ based on the robust optimization  techniques reviewed in  \textbf{Robust  Optimization with Constraints}. 

Firstly, we consider the $i$th row of the constraint  \eqref{equ: constraint on tilde x_t}, i.e.
$D_{x,i}^\top  \tilde x_{t+1}\leq d_{x,i}-\epsilon_1$ $ \forall\, \{w_k \in \W\}_{k=0}^T,$
where $D_{x,i}^\top$ denotes the $i$th row of the matrix $D_x$. This constraint is equivalent to 
$ \sup_{\{w_k \in \W\}_{k=0}^T}(D_{x,i}^\top  \tilde x_{t+1})\leq d_{x,i}-\epsilon_1$. Further, by  \eqref{equ: robust opt L1} and the definitions of $\tilde x_{t+1}$ and $\W$, we obtain
\begin{align*}
 \sup_{\{w_k \!\in \!\W\}} \!D_{x,i}^\top \tilde x_{t+1}\!&=\! \sup_{\{w_k \in \W\}} D_{x,i}^\top \sum_{s=1}^{2H}\Phi^x_{s}(\bm M_{t-H+1:t})  w_{t\!+\!1\!-\!s}\\
& = \!\sum_{s=1}^{2H} \!\sup_{w_{t\!+\!1\!-\!s}\in\W}\! D_{x,i}^\top \Phi^x_{s}(\bm M_{t-H+1:t}) w_{t\!+\!1\!-\!s}\\
& = \sum_{s=1}^{2H}  \|D_{x,i}^\top \Phi^x_{s}(\bm M_{t-H+1:t} )\|_1 \bar w
\end{align*}
Define $g_i^x(\bm M_{t-H+1:t})= \sum_{s=1}^{2H}  \|D_{x,i}^\top \Phi^x_{s}(\bm M_{t-H+1:t} )\|_1 \bar w.$ Hence, the robust constraint \eqref{equ: constraint on tilde x_t} on $\tilde x_{t+1}$ is equivalent to the following polytopic constraints on $\bm M_{t-H+1:t}$:
\begin{align}\label{equ: temporal constraints x}
g_i^x(\bm M_{t-H+1:t}) \leq d_{x,i}-\epsilon_1, \quad \forall \, 1\leq i \leq k_x.
\end{align}

Similarly, the constraint \eqref{equ: constraint on tilde u_t} on $\tilde u_t$ is equivalent to:
\begin{align}\label{equ: temporal constraints u}
g_j^u(\bm M_{t-H:t}) \leq d_{u,j}-\epsilon_1, \quad \forall \, 1\leq j \leq k_u,
\end{align}
where  $g_j^u(\bm M_{t-H:t})= \sum_{s=1}^{2H}  \|D_{u,j}^\top \Phi^u_{s}(\bm M_{t-H:t} )\|_1 \bar w.$

\subsubsection{Step 3: OCO with Memory and Temporal-coupled Constraints} By Step 2 and our review of robust optimization, we can convert the constrained online optimal control problem into \textit{OCO with memory and coupled constraints}. That is, at each  $t$, the decision maker selects a policy $\bm M_t$ satisfying constraints \eqref{equ: temporal constraints x} and \eqref{equ: temporal constraints u}, and then incurs a cost $f_t(\bm M_{t-H:t})$. In our framework, both the constraints \eqref{equ: temporal constraints x}, \eqref{equ: temporal constraints u} and the cost function $f_t(\bm M_{t-H:t})$ couple the current policy with the historical policies. This makes the problem far more challenging than OCO with memory which only considers coupled costs \cite{anava2015online}. 

\subsubsection{Step 4: Benefits of the Slow Variation of  Online Policies} 
We approximate the coupled constraint functions $g^x_i(\bm M_{t-H+1:t})$ and $g^u_j(\bm M_{t-H:t})$ as   decoupled ones below,
\begin{align*}
&\mathring g^x_i(\bm M_t)\!=\!g^x_i(\bm M_{t}, \dots, \bm M_t), \mathring g^u_i(\bm M_t)\!=\!g^u_i(\bm M_{t}, \dots, \bm M_t),
\end{align*}
by letting the historical policies $\bm M_{t-H:t-1}$  be identical to the current  $\bm M_t$.\footnote{Though we consider $\bm M_t=\dots=\bm M_{t-H}$ here, the component $M_t^{[i]}$ of $\bm M_t=\{M_t^{[i]}\}_{i=1}^H$ can be  different for different $i$.}
If the online policy $\bm M_t$ varies slowly with $t$, which is  satisfied by most OCO algorithms (e.g.  OGD with a diminishing stepsize \cite{hazan2019introduction}), one may be able to bound the approximation errors by
 $g^x_i(\bm M_{t-H+1:t})- \mathring g^x_i(\bm M_t) \leq \epsilon_2$  and $g^u_j(\bm M_{t-H:t})- \mathring g^u_j(\bm M_t) \leq \epsilon_2$  for a small $\epsilon_2>0$. Consequently,  the  constraints \eqref{equ: temporal constraints x} and \eqref{equ: temporal constraints u}
are ensured by the  polytopic constraints that only depend on $\bm M_t$:
\begin{align}
&\mathring g^x_i(\bm M_t)\leq \!d_{x,i}\!-\epsilon_1-\epsilon_2,    \ \mathring g^u_j(\bm M_t)\leq \!d_{u,j}\!-\epsilon_1-\epsilon_2,
\end{align}
where the buffer zone $\epsilon_2$  allows for the approximation error caused by neglecting the  variation of the online policies.

\subsubsection{Step 5: Conversion to  OCO}
By Step  4, we  define a  decoupled search space/constraint set on each policy below,
\begin{equation}\label{def: Omega_epsilon}
	\begin{aligned} 
\Omega_\epsilon\!= \!\{ \bm M\in \M:&\ \mathring g^x_i(\bm M) \leq d_{x,i}-\epsilon, \forall 1\leq i\leq k_x,\\
&\mathring g^u_j(\bm M)\leq d_{u,j}-\epsilon,  \forall\,  1\leq j \leq k_u \},
\end{aligned}
\end{equation}
where $\M$ is a bounded convex constraint set  defined as 
$$\M =\{ \bm M: \ \|M^{[i]}\|_\infty \leq 2\sqrt n \kappa^3 (1-\gamma)^{i-1}, \ \forall \, 1\leq i \leq H\}.$$
Our set $\M$ is slightly different from $\M_2$ in \eqref{equ: M_2 def} to ensure that   $\Omega_\epsilon$ is a polytope.\footnote{Compared with  $\M_2$,  our $\M$ uses the $L_\infty$ norm;  the  $\sqrt n$ factor accounts for the change of  norms; and  $\kappa_B$ disappears because we can   prove that  $\kappa_B$ is not necessary here (see \cite{li2020online}).} 
Notice that $\Omega_\epsilon$ provides buffer zones with size $\epsilon$ to account for the  approximation errors $\epsilon_1$ and   $\epsilon_2$.  
Based on $\Omega_\epsilon$ and  technique \eqref{equ: mathring f_t(M_t)}, we can further convert the ``OCO with memory and coupled constraints'' in Step 3 into a classical OCO problem below. That is,  at each  $t$, the agent selects a policy $\bm M_t\in \Omega_\epsilon$, and then suffers a convex stage cost $\mathring f_t(\bm M_t)$  defined in \eqref{equ: mathring f_t(M_t)}.  We apply online gradient descent to solve this OCO problem, as described in Algorithm \ref{alg:ogd}. We select the stepsizes of OGD  to be  small enough to ensure small approximation errors from Step 4 and thus small buffer zones,  but also  to be large enough to allow online policies to adapt to time-varying environments. Conditions for suitable stepsizes are discussed in \textbf{Theoretical Results}.

\begin{algorithm2e}
	\caption{OGD-BZ}
	\label{alg:ogd}
	\SetAlgoNoLine
	\DontPrintSemicolon
	\LinesNumbered
	\KwIn{A $(\kappa,\gamma)$-strongly stable  matrix $\mathbb K$, parameter  $H>0$,  buffer size $\epsilon$,   stepsize $\eta_t$. }
	
	Determine the polytopic constraint set  $\Omega_\epsilon$ by \eqref{def: Omega_epsilon} with buffer size $\epsilon$ and	 initialize $\bm M_0 \in \Omega_\epsilon$.\;

	\For{$t=0,  1,2,\ldots,T$}{
		Implement action $	u_t=-\Kb x_t \!+\! \sum_{i=1}^H M_t^{[i]}w_{t-i}.
		$\;
		
		Observe the next state $x_{t+1}$ and record $w_t=x_{t+1}-A x_t-B u_t$.\;
		
		Run projected OGD 
		$$ \bm M_{t+1}=\Pi_{\Omega_\epsilon}\left[\bm M_t-\eta_t \nabla \mathring f_t(\bm M_t)\right]$$
		where $\mathring f_t(\bm M)$ is defined in \eqref{equ: mathring f_t(M_t)}.}
	
\end{algorithm2e}

In Algorithm \ref{alg:ogd}, the most computationally demanding step at each stage is the projection onto the polytope $\Omega_\epsilon$, which requires solving a  quadratic program. Nevertheless, one can reduce the online computational burden via offline computation by leveraging the solution structure of quadratic programs (see \cite{alessio2009survey} for more details).  

Lastly, we  note that other OCO algorithms can be applied to solve this problem too, e.g. online natural gradient, online mirror descent, etc. One can also apply projection-free methods, e.g. \cite{yuan2018online}, to reduce the computational burden at the expense of $o(T)$ constraint violation.

\begin{remark}
	To ensure safety,   safe RL literature usually constructs a safe set for the state \cite{fisac2018general}, while this paper constructs a safe search space $\Omega_\epsilon$ for the policies directly.    Besides, safe RL literature may employ unsafe policies occasionally, for example, \citet{fisac2018general} allows unsafe exploration policies within the safe set and changes to a safe policy on the boundary of the safe set. However,  our search space $\Omega_\epsilon$ only contains safe policies. Despite a smaller policy search space, our OGD-BZ still achieves desirable (theoretical) performance. Nevertheless, when the system is unknown, larger sets of exploration policies may benefit the  performance, which is left as future work.

\end{remark}

\begin{remark}
    It is  worth comparing our method with a well-known robust MPC method: tube-based robust MPC (see e.g. \citet{rawlings2009model}). Tube-based robust MPC also tightens the constraints to allow for model inaccuracy and/or disturbances. However, tube-based robust MPC considers constraints on the states, while our method converts the state (and action) constraints into the constraints on the  policy parameters by leveraging the properties of  disturbance-action policies.
\end{remark}


\section{Theoretical 
Results}\label{sec: theory}
In this section, we show  that OGD-BZ guarantees both safety and   $\tilde O(\sqrt T)$ policy regret under proper   parameters.

\subsubsection{Preparation} 
To establish the conditions on the parameters for our theoretical results, we  introduce three quantities $\epsilon_1(H), \epsilon_2(\eta, H), \epsilon_3(H)$ below.  We note that $\epsilon_1(H)$ and $\epsilon_2(\eta, H)$ bound the approximation errors in Step 1 and Step 4 of the previous section respectively (see Lemma \ref{lem: bdd xt-H part} and Lemma \ref{lem: g(Mt:t-H) to g(Mt) bdd error} in the proof of Theorem \ref{thm: feasibility} for more details). $\epsilon_3(H)$ bounds the  constraint violation of the disturbance-action policy  $\bm M(K)$, where $\bm M(K)$ approximates the linear controller $u_t=-K x_t$ for any $K\in \K$ (see Lemma \ref{lem: define M_ap and epsilon3} in the proof of Theorem \ref{thm: feasibility} for more details). 
\begin{definition}\label{def: epsilon def}
	We define 
	\begin{align*}
	\epsilon_1(H)&=c_1 n\sqrt m H (1-\gamma)^H,	\epsilon_2(\eta, H)= c_2\eta \cdot n^2 \sqrt m H^2\\
	\epsilon_3(H)&=c_3\sqrt n (1-\gamma)^H
	\end{align*}
	where $c_1$, $c_2$, and $c_3$  are polynomials of   $\|D_x\|_\infty, \|D_u\|_\infty$, $ \kappa, \kappa_B, \gamma^{-1}, \bar w, G$.
\end{definition}

\subsubsection{Safety of OGD-BZ}

\begin{theorem}[Feasibility \& Safety]\label{thm: feasibility}
	
	Consider  constant stepsize $\eta_t=\eta$, $\epsilon\geq 0,H\geq \frac{\log(2\kappa^2)}{\log((1-\gamma)^{-1})}$. If 
	the buffer size $\epsilon$ and $H$ satisfy
	$$\epsilon\leq \epsilon_*-\epsilon_1(H)-\epsilon_3(H),$$
	the set $\Omega_\epsilon$ is non-empty. 
	Further,  if $\eta$, $\epsilon$ and $H$ also satisfy
	$$	\epsilon\geq \epsilon_1(H)+\epsilon_2(\eta, H),$$
	our OGD-BZ is safe, i.e. $x^{\textup{OGD-BZ}}_t \in \X$ and $u_t^{\textup{OGD-BZ}} \in \U$ for all $t$ and for any  disturbances $\{w_k\in \W\}_{k=0}^T$. 
\end{theorem}

\nbf{Discussions:} Firstly, Theorem \ref{thm: feasibility} shows that $\epsilon$ should be small enough to ensure a nonempty $\Omega_\epsilon$ and thus valid/feasible outputs of OGD-BZ. This  is intuitive since the constraints are more conservative as $\epsilon$ increases. Since $\epsilon_1(H)+\epsilon_3(H)=\Theta(H(1-\gamma)^H)$ decays with $H$ by Definition \ref{def: epsilon def}, the first condition also implies a large enough $H$.

Secondly, Theorem \ref{thm: feasibility} shows that, to ensure safety, the buffer size $\epsilon$ should also be large enough to allow for the total approximation errors $\epsilon_1(H)+\epsilon_2(\eta, H)$, which is consistent with our discussion in the previous section. To ensure the compatibility of the two conditions on $\epsilon$, the approximation errors $\epsilon_1(H)+\epsilon_2(\eta, H)$ should be small enough, which  requires a large enough $H$ and a small enough $\eta$ by Definition \ref{def: epsilon def}. 

In conclusion, the safety requires a large enough $H$, a small  enough $\eta$, and an $\epsilon$ which is neither too large nor too small. For example, we can select $\eta\leq \frac{\epsilon_*}{8 c_2 n^2 \sqrt m H^2}$,   $\epsilon_*/4 \leq \epsilon\leq 3\epsilon_*/4$, and $H\geq \max(\frac{\log(\frac{8(c_1+c_3)n\sqrt m}{\epsilon_*}T)}{\log((1-\gamma)^{-1}},  \frac{\log(2\kappa^2)}{\log((1-\gamma)^{-1})})$.

\begin{remark}\label{remark: infinite horizon safe}
It can be shown that it is safe to implement any $\bm M\in \Omega_{\epsilon}$ for an infinite horizon ($0\leq t \leq +\infty$) under the conditions of Theorem \ref{thm: feasibility} based on the proof of Theorem \ref{thm: feasibility}. For more details, please refer to the end of the proof of Theorem \ref{thm: feasibility}.
\end{remark}



\subsubsection{Policy Regret Bound for OGD-BZ}

\begin{theorem}[Regret Bound]\label{thm: regret bdd} 
	Under the conditions in Theorem \ref{thm: feasibility}, OGD-BZ  enjoys the regret bound below:
\begin{align*}
 \textup{Reg}(\textup{OGD-BZ})\leq  \,&  O\!\left(n^3 m H^3\eta T+\frac{mn}{\eta}\right.\\
&+ \!(1\!-\!\gamma)^H\! H^{2.5}T ({n^{ {3}}m^{ {1.5}}}\!+\!\sqrt{k_c} m n^{2.5})/\epsilon_*\\
&+\left.\epsilon T H^{1.5}({n^{ {2}}m}+\sqrt{k_cm n^3})/\epsilon_*\right),
\end{align*}
	where  the hidden constant depends polynomially on $\kappa, \kappa_B,$ $ \gamma^{-1},  {\|D_x\|_\infty, \|D_u\|_\infty,} \|d_x\|_2, \|d_u\|_2, \bar w, G$.

\end{theorem}
Theorem \ref{thm: regret bdd} provides a regret bound for OGD-BZ as long as OGD-BZ is safe. Notice that as the buffer size $\epsilon$ increases, the regret bound becomes worse. This is intuitive since our OGD-BZ will have to search for policies in a smaller set $\Omega_{\epsilon}$ if $\epsilon$ increases. Consequently, the buffer size $\epsilon$ can serve as a tuning parameter for the trade-off between safety and regrets, i.e. a small $\epsilon$ is preferred for low regrets while a large $\epsilon$ is preferred for safety (as long as $\Omega_\epsilon\not=\emptyset$).
In addition, although a small stepsize $\eta$ is preferred for safety in Theorem \ref{thm: feasibility}, Theorem \ref{thm: regret bdd}  suggests that the stepsize should not be  too small for low regrets since the regret bound contains a $\Theta(\eta^{-1})$ term. This is intuitive since  the stepsize $\eta$ should be large enough to allow OGD-BZ to adapt to the varying objectives for better online performance.

Next, we provide a regret bound with specific   parameters.

\begin{corollary}\label{cor: regret bdd}
For sufficiently large $T$, when   $H\geq  \frac{ \log(8(c_1+c_2)n\sqrt mT/\epsilon_*)}{\log((1-\gamma)^{-1})}$, $\eta=\Theta(\frac{1}{n^2\sqrt m H \sqrt T})$,  $\epsilon=\epsilon_1(H)+\epsilon_2(\eta, H)=\Theta(\frac{\log(n \sqrt{m}T)}{\sqrt{T}})$, OGD-BZ is safe and
$
	\textup{Reg}(\textup{OGD-BZ})\leq \tilde O\left((n^{ {3}}m^{1.5}{k_c}^{0.5})\sqrt T \right).
$
\end{corollary}
Corollary \ref{cor: regret bdd} shows that OGD-BZ achieves  $\tilde O(\sqrt T)$ regrets when $H\geq\Theta(\log T)$, $\eta^{-1}=\tilde \Theta(\sqrt T)$, and $\epsilon=\tilde \Theta(1/\sqrt T)$. This demonstrates that OGD-BZ can ensure both constraint satisfaction and sublinear regrets under the proper parameters of the algorithm. We remark that  a larger $H$ is preferred for better performance due to smaller approximation errors and a potentially larger policy search space $\Omega_\epsilon$, but the computational complexity of OGD-BZ increases with $H$. 
Besides, though the  choices of $H$, $\eta$, and $\epsilon$ above   require the prior knowledge of $T$,  one can apply doubling tricks \cite{hazan2019introduction} to avoid this requirement. Lastly, we note that our $\tilde O(\sqrt T)$ regret bound is consistent with the unconstrained online optimal control literature for convex cost functions \cite{agarwal2019online}. For strongly convex costs, the regret for the unconstrained case is logarithmic in $T$ \cite{agarwal2019logarithmic}.  We leave the study on the constrained control with strongly convex costs for the future.

\subsection{Proof of Theorem \ref{thm: feasibility}}\label{subsec: feasible proof}

To prove Theorem \ref{thm: feasibility}, we first provide lemmas to bound errors by $\epsilon_1(H), \epsilon_2(\eta,H)$, and $\epsilon_3(H)$, respectively. The proofs of  Lemmas \ref{lem: bdd xt-H part}-\ref{lem: define M_ap and epsilon3} and Corollary \ref{cor: K's M in Omega} in this subsection are provided in the arXiv version~\cite{li2020online}.

Firstly, we show that the approximation error in Step 1 of the previous section can be bounded by $\epsilon_1(H)$. 
\begin{lemma}[Error bound $\epsilon_1(H)$]\label{lem: bdd xt-H part}
	When  $\bm M_k \in \M$ for all $k$ and  $H\geq \frac{\log(2\kappa^2)}{\log((1-\gamma)^{-1})}$, we have
	\begin{align*}
	&\max_{\|w_k\|_\infty \leq \bar w} \|D_{x} \AK^H x_{t-H}\|_\infty\leq \epsilon_1(H),\\
	& \max_{\|w_k\|_\infty \leq \bar w} \|D_{u}\Kb \AK^H x_{t-H}\|_\infty \leq \epsilon_1(H).
	\end{align*}

\end{lemma}

The proof of Lemma \ref{lem: bdd xt-H part} relies on the boundedness of $x_t$ when implementing $\bm M_t\in \mathcal M$ as stated below.

\begin{lemma}[Bound on $x_t$]\label{lem: bdd on xt}
    With $\bm M_k \in \mathcal M$ for all $k$ and $\kappa^2(1-\gamma)^H<1$, we have
    $$ \|x_t\|_2\leq b,$$
    where $b=\frac{\kappa\sqrt{n}\bar w (\kappa^2+2\kappa^5 \kappa_B \sqrt{mn} H)}{(1-\kappa^2(1-\gamma)^H)\gamma}+2\sqrt{mn}\kappa^3\bar w /\gamma$. Hence, when $H\geq \frac{\log(2\kappa^2)}{\log((1-\gamma)^{-1})}$,  we have $b \leq 8 \sqrt{mn^2} H \bar w \kappa^6 \kappa_B/\gamma$.
\end{lemma}

Secondly, we show that   the  error  incurred by the Step 3 of the previous section can be bounded by $\epsilon_2(\eta,H)$. 
\begin{lemma}[Error bound $\epsilon_2(\eta, H)$]\label{lem: g(Mt:t-H) to g(Mt) bdd error}
 When $H\geq \frac{\log(2\kappa^2)}{\log((1-\gamma)^{-1})}$, 
	the policies $\{\bm M_t\}_{t=0}^T$ generated by OGD-BZ with a constant stepsize $\eta$ satisfy
	\begin{align*}
	&\max_{1\leq i \leq k_x} |\mathring	g^x_i(\bm M_t)-	g^x_i(\bm M_{t-H+1:t})| \leq \epsilon_2(\eta,H),\\
	&\max_{1\leq j \leq k_u} |	\mathring g^u_j(\bm M_t)-	g^u_j(\bm M_{t-H:t})| \leq \epsilon_2(\eta,H).
	\end{align*}

\end{lemma}

Thirdly, we show that for any $K\in \K$, there  exists a disturbance-action policy $\bm M(K)\in \M$ to approximate the policy $u_t=-Kx_t$. However, $\bm M(K)$ may not be safe and is only $\epsilon_3(H)$-loosely safe. 
\begin{lemma}[Error bound $\epsilon_3( H)$]\label{lem: define M_ap and epsilon3}

	For any $K \in \K$, there exists a disturbance-action policy $\bm M(K)=\{M^{[i]}(K)\}_{i=1}^H\in \M$ defined as
	$M^{[i]}(K)=(\Kb-K)(A-BK)^{i-1}$
	such that
	\begin{align*}
	  \max(\|{D_x}[x_t^K \!-\! x_t^{\bm M(K)}\!]\|_{{\infty}}, \|{D_u}[u_t^K\! -\! u_t^{\bm M(K)}\!]\|_{{\infty}}) \!\leq\! \epsilon_3(H)
	\end{align*}
	where $(x_t^K, u_t^K)$ and $(x_t^{\bm M(K)}, u_t^{\bm M(K)})$ are produced by controller $u_t\!=\!-K x_t$ and disturbance-action policy $\bm M(K)$ respectively. 
		Hence,  $\bm M(K)$ is  $\epsilon_3(H)$-loosely safe.

\end{lemma}

Based on Lemma \ref{lem: define M_ap and epsilon3}, we can further show that $\bm M(K)$ belongs to  a polytopic constraint set in the following corollary. For the rest of the paper, we will omit the arguments in $\epsilon_1(H), \epsilon_2(\eta, H), \epsilon_3(H)$ for notational simplicity.
\begin{corollary}\label{cor: K's M in Omega}
Consider  $K\in \K$, if $K$  is  $\epsilon_0$-strictly safe for $\epsilon_0\geq 0$,  then  
$
	\bm M(K)\in \Omega_{\epsilon_0-\epsilon_1-\epsilon_3}.
$
\end{corollary}

\begin{proof}[Proof of Theorem \ref{thm: feasibility}] For notational simplicity, we denote the states and actions generated by OGD-BZ as $x_t$ and $u_t$ in this proof.
    First, we show  $\bm M(K_*)\in \Omega_\epsilon$ below. 
Since $K_*$ defined in Assumption \ref{ass: K strictly feasible} is $\epsilon_*$-strictly safe,  by Corollary \ref{cor: K's M in Omega},  there exists $\bm M(K_*)\in \Omega_{\epsilon_*-\epsilon_1-\epsilon_3}$. 
Since the set $\Omega_\epsilon$ is smaller as $\epsilon$ increases, 
when $\epsilon_*-\epsilon_1-\epsilon_3\geq \epsilon$, we have $\bm M(K_*) \in \Omega_{\epsilon_*-\epsilon_1-\epsilon_3}\subseteq \Omega_\epsilon$, so  $\Omega_\epsilon$ is non-empty.

Next, we prove the safety   by  Lemma \ref{lem: bdd xt-H part} and Lemma \ref{lem: g(Mt:t-H) to g(Mt) bdd error} based on the discussions in the previous section. Specifically, OGD-BZ guarantees that $\bm M_t\in \Omega_\epsilon$ for all $t$. Thus, by Lemma \ref{lem: g(Mt:t-H) to g(Mt) bdd error}, we have $g_i^x(\bm M_{t-H:t-1})= g_i^x(\bm M_{t-H:t-1})- \mathring g_i^x(\bm M_{t-1})+ \mathring g_i^x(\bm M_{t-1}) \leq \epsilon_2+d_{x,i}-\epsilon$ for any $i
$.   Further, by Step 2 of the previous section and Lemma \ref{lem: bdd xt-H part}, we have $D_{x,i}^\top x_t =  D_{x,i}^\top \AK^H x_{t-H}\!+ \! D_{x,i}^\top \tilde x_{t}\leq \|D_x \AK^H x_{t-H}\|_\infty + g_i^x(\bm M_{t-H:t-1}) \leq \epsilon_1 +\epsilon_2+d_{x,i}-\epsilon \leq d_{x,i}$  for any $\{w_k\in \W\}_{k=0}^T$ if $\epsilon\geq \epsilon_1+\epsilon_2$. Therefore, $x_t\in \X$ for all $w_k\in \W$. Similarly, we can show $u_t\in \U$ for any $w_k\in \W$. Thus, OGD-BZ is safe.
\end{proof}

\noindent\textit{Proof of Remark \ref{remark: infinite horizon safe}.}
When implementing  $\bm M\in \Omega_\epsilon$ for an infinite horizon, we have $\mathring g_i^x(\bm M)=g_i^x(\bm M, \dots, \bm M)\leq d_{x,i}-\epsilon$. Since Proposition \ref{prop: approx} holds for any $t\geq 0$, we still have $D_{x,i}^\top x_t =  D_{x,i}^\top \AK^H x_{t-H}\!+ \! D_{x,i}^\top \tilde x_{t}\leq \|D_x \AK^H x_{t-H}\|_\infty + g_i^x(\bm M,\dots, \bm M) \leq \epsilon_1 +d_{x,i}-\epsilon \leq d_{x,i}$  for any $\{w_k\in \W\}_{k=0}^T$ if $\epsilon\geq \epsilon_1+\epsilon_2$. Thus, $x_t\in \X$ for all $w_k\in \W$ for $t\geq 0$. Constraint satisfaction of $u_t$ can be proved similarly.


\subsection{Proof of Theorem \ref{thm: regret bdd}}

We divide the regret into three parts and bound each part.
\begin{align*}
&\ \text{Reg}(\text{OGD-BZ})=  \  J_T(\A)-\min_{K \in \K} J_T(K)\\
=& \ \underbrace{J_T(\A)- \sum_{t=0}^T \mathring f_t(\bm M_t)}_{\text{Part i}}+ \underbrace{\sum_{t=0}^T \mathring f_t(\bm M_t)- \min_{\bm M\in \Omega_\epsilon}\sum_{t=0}^T \mathring f_t(\bm M)}_{\text{Part ii}}\\
& \ +\underbrace{ \min_{\bm M\in \Omega_\epsilon} \sum_{t=0}^T\mathring f_t(\bm M)-\min_{K \in \K} J_T(K)}_{\text{Part iii}}
\end{align*}

\nbf{Bound on Part ii.} Firstly, we bound Part ii based on the regret bound of OGD in the literature \cite{hazan2019introduction}.
\begin{lemma}\label{lem: part ii}
	With a constant stepsize $\eta$, we have
$
	\textup{Part ii} \leq \delta^2/{2\eta}+ {\eta}G_f^2 T/2
$,
	where $\delta= 	\sup_{\bm M, \tilde{\bm M} \in \Omega_\epsilon} \|\bm M-\tilde{\bm M}\|_F \leq  4 \sqrt{mn} \kappa^3/\gamma$ and $G_f=\max_{t}\sup_{\bm M\in \Omega_\epsilon}\|\nabla \mathring f_t(\bm M)\|_F\leq\Theta( Gb(1+\kappa)\sqrt n \bar w\kappa^2\kappa_B \sqrt H \frac{1+\gamma}{\gamma} )$. Consequently, when $H\geq \frac{\log(2\kappa^2)}{\log((1-\gamma)^{-1})}$, we have $G_f\leq \Theta(\sqrt{n^3H^3 m})$ and the hidden factor is quadratic on $\bar w$.
\end{lemma}
 The  proof details are provided in \cite{li2020online}.

\vspace{1pt}

\nbf{Bound on Part iii.}  For notational simplicity, we denote 
$\bm M^*=\argmin_{\Omega_\epsilon}\sum_{t=0}^T\mathring f_t(\bm M), K^*=\argmin_{\K} J_T(K).$
By Lemma \ref{lem: define M_ap and epsilon3}, we can construct a loosely safe $\bm M_{\text{ap}}=\bm M(K^*)$ to approximate $K^*$.  By Corollary \ref{cor: K's M in Omega}, we have
\begin{align}\label{equ: M_ap non feasible}
    \bm M_{\text{ap}} \in \Omega_{-\epsilon_1-\epsilon_3}.
\end{align} 
We will bound Part iii by leveraging $\bm M_{\text{ap}}$ as middle-ground and bounding the Part iii-A and Part iii-B defined below.
\begin{align*}
    \text{Part iii}=\! \underbrace{\sum_{t=0}^T \!(\mathring f_t(\bm M^*\!)\!-\!  \mathring f_t(\bm M_{\text{ap}}))}_{\text{Part iii-A}}\!+\!\underbrace{\sum_{t=0}^T \!\mathring f_t(\bm M_{\text{ap}})\!-\!J_T(K^*\!)}_{\text{Part iii-B}}
\end{align*}

\begin{lemma}\label{lem: part iii second half}
		Consider  $K^* \in \K$ and $\bm M_{\text{ap}}=\bm M(K^*)$, then
$
	\textup{Part iii-B}\leq \Theta(Tn^2 m H^2(1-\gamma)^H).
$
\end{lemma}
\begin{lemma}\label{lem: part iii first half}
		Under the conditions in Theorem \ref{thm: regret bdd}, we have
	\begin{align*}
	\textup{Part iii-A}&\leq \Theta\left(\! (\epsilon_1+\epsilon_3+\epsilon)T H^{\frac{3}{2}}\frac{{n^{{2}}m}\!+\!\sqrt{k_cm n^3}}{\epsilon_*}\!\right).
	\end{align*}
\end{lemma}
We highlight that  $\bm M_{\text{ap}}$ may not belong to $\Omega_{\epsilon}$ by \eqref{equ: M_ap non feasible}. Therefore, even though $\bm M^*$ is optimal in $\Omega_{\epsilon}$, Part iii-A can still be positive and has to be bounded to yield a regret bound. This is  different from the unconstrained online  control literature \cite{agarwal2019logarithmic}, where Part iii-A is non-positive because $\bm M_{\text{ap}}\in \mathcal M$ and $\bm M^*$ is  optimal  in the same set $\mathcal M$ when there are no constraints (see \cite{agarwal2019logarithmic} for more details). 

\vspace{1pt}

\nbf{Bound on Part i.}  Finally, we provide a bound on Part i.
\begin{lemma}\label{lem: part i}
	Apply Algorithm \ref{alg:ogd} with constant stepsize $\eta$, then
$
	\textup{Part i} \!\leq\! O(T n^2m H^2(1-\gamma)^H\!\! +\!n^3 m H^3\eta T  )
$.
\end{lemma}
The proofs of Lemma \ref{lem: part iii second half} and Lemma \ref{lem: part i} are similar to those in \citet{agarwal2019logarithmic}.

Finally,  Theorem \ref{thm: regret bdd} can be proved by summing up the bounds on Part i, Part ii, Part iii-A, and Part iii-B in Lemmas \ref{lem: part ii}-\ref{lem: part i} and  only explicitly showing the highest order terms. 

\subsection{Proof of Lemma \ref{lem: part iii first half}}
We define
$ \bm M^\dagger=\argmin_{\Omega_{-\epsilon_1-\epsilon_3}} \sum_{t=0}^T \mathring f_t(\bm M)$.
By \eqref{equ: M_ap non feasible},  we have $\sum_{t=0}^T \mathring f_t(\bm M_{\text{ap}})\geq\sum_{t=0}^T \mathring f_t(\bm M^\dagger)$. Therefore, it suffices to bound
$
 \sum_{t=0}^T \mathring f_t(\bm M^*)-\sum_{t=0}^T \mathring f_t(\bm M^\dagger)
$, which can be viewed as the difference in the optimal values when perturbing the feasible/safe set from $\Omega_\epsilon$ to $\Omega_{-\epsilon_1-\epsilon_3}$. To bound Part iii-A, we establish a perturbation result  by leveraging the polytopic structure of  $\Omega_{\epsilon}$ and $\Omega_{-\epsilon_1-\epsilon_3}$.
\begin{proposition}\label{prop: perturbation}
	Consider two polytopes
	$\Omega_1=\{x: Cx \leq h\}$, $\Omega_2=\{x: Cx \leq h-\Delta\}$, where $\Delta_i\geq 0$ for all $i$. Consider a convex function $f(x)$ that is $L$-Lipschitz continuous on $\Omega_1$.  If $\Omega_1$ is bounded, i.e. $\sup_{x_1, x_1'\in \Omega_1}\|x_1-x_1'\|_2\leq \delta_1$ and if $\Omega_2$ is non-empty, i.e. there exists $\mathring x\in \Omega_2$, then 
	\begin{align}\label{equ: perturbation}
	|\min_{\Omega_1} f(x)- \min_{\Omega_2} f(x)| \leq  \frac{L\delta_1 \|\Delta\|_\infty}{\min_{\{i: \Delta_i>0\}} (h-C \mathring x)_i}.
	\end{align}
\end{proposition}

To prove Lemma \ref{lem: part iii first half}, it suffices to bound the quantities in \eqref{equ: perturbation}  for our problem and then plug them in \eqref{equ: perturbation}.

\begin{lemma}\label{lem: diameter and L bound}
	There exists an enlarged polytope $\Gamma_{\epsilon}=\{\vec W: C \vec W\leq h_\epsilon\}$ that is equivalent to $\Omega_{\epsilon}$ for any $\epsilon\in \R$, where $\vec W$ contains elements of $\bm M$ and  auxiliary variables.
	
	Further,  under the conditions of Theorem \ref{thm: feasibility}, (i) $\Gamma_{-\epsilon_1-\epsilon_3}$ is bounded by 
	$ \delta_1=\Theta(\sqrt{m n}+ \sqrt{k_c})$; (ii) $\sum_{t=0}^T \mathring f_t(\bm M)$   is Lipschitz continuous with $L=\Theta(T (nH)^{1.5}\sqrt m)$; (iii) the difference $\Delta$   between $\Gamma_{\epsilon}$ and $\Gamma_{-\epsilon_1-\epsilon_3}$ 
	satisfies $\|\Delta\|_\infty = \epsilon+\epsilon_1+\epsilon_3$; (iv) there exists ${\vec W^{\circ}}\in \Gamma_{\epsilon}$  s.t. $\min_{\{i:\Delta_i>0\}} (h_{(-\epsilon_1-\epsilon_3)}\!-\!C\vec W^{\circ})_i\geq \epsilon_*$.

\end{lemma}


\section{Numerical Experiments}

In this section, we numerically test our OGD-BZ on a thermal control problem with a Heating Ventilation and Air
Conditioning (HVAC) system. 
Specifically, we consider the linear thermal dynamics studied in \citet{zhang2016decentralized} with additional random disturbances, that is, 
$\dot x(t)= \frac{1}{\upsilon\zeta}(\theta^{\mathrm{o}}(t)-x(t))- \frac{1}{\upsilon} u(t)+ \frac{1}{\upsilon} \pi + \frac{1}{\upsilon} w(t)
$,
where $x(t)$ denotes the  room temperature  at time $t$, $u(t)$ denotes the control input  that is related with the air flow rate of the HVAC system, $\theta^{\mathrm{o}}(t)$ denotes the outdoor temperature, $w(t)$ represents random disturbances, $\pi$ represents external  heat sources' impact,  $\upsilon$ and $\zeta$ are physical constants.  We discretize the thermal dynamics with  $\Delta_t=60$s. For human comfort and/or safe operation of device, we impose constraints on the room temperature, $x(t)\in [x_{\min}, x_{\max}]$, and the control inputs, $u(t)\in [u_{\min}, u_{\max}]$.
Consider  a desirable temperature $\theta^{set}$ set by the user and a control setpoint $u^{set}$. Consider the cost function  
$c(t)=q_{t}(x(t) - \theta^{set})^2 + r_{t} (u(t)-u^{set})^2$.

In our  experiments, we consider $v=100$, $\zeta =6$, $\theta^o=30^\circ\text{C}$, $\pi=1.5$, and let $w_t$ be i.i.d. generated from $\text{Unif}(-2,2)$. Besides, we consider  $\theta^{set}=24^\circ\text{C}$, $x_{\min}=22^\circ\text{C}$, $x_{\max}=26^\circ\text{C}$, $u_{\min}=0$, $u_{\max}=5$. We consider $q_t=2$ for all $t$ and time-varying $r_t$ generated i.i.d. from $\text{Unif}(0.1,4)$. When applying OGD-BZ, we select $H=7$ and a diminishing stepsize $\eta_t=\Theta(t^{-0.5})$, i.e. we let $\eta_t=0.5 (40)^{-0.5}$ for $t< 40$ and  $\eta_t=0.5 (t+1)^{-0.5}$ for $t\geq 40$.

\begin{figure*}
\centering
    \subfigure[Averaged regret]{\includegraphics[width=0.3\linewidth]{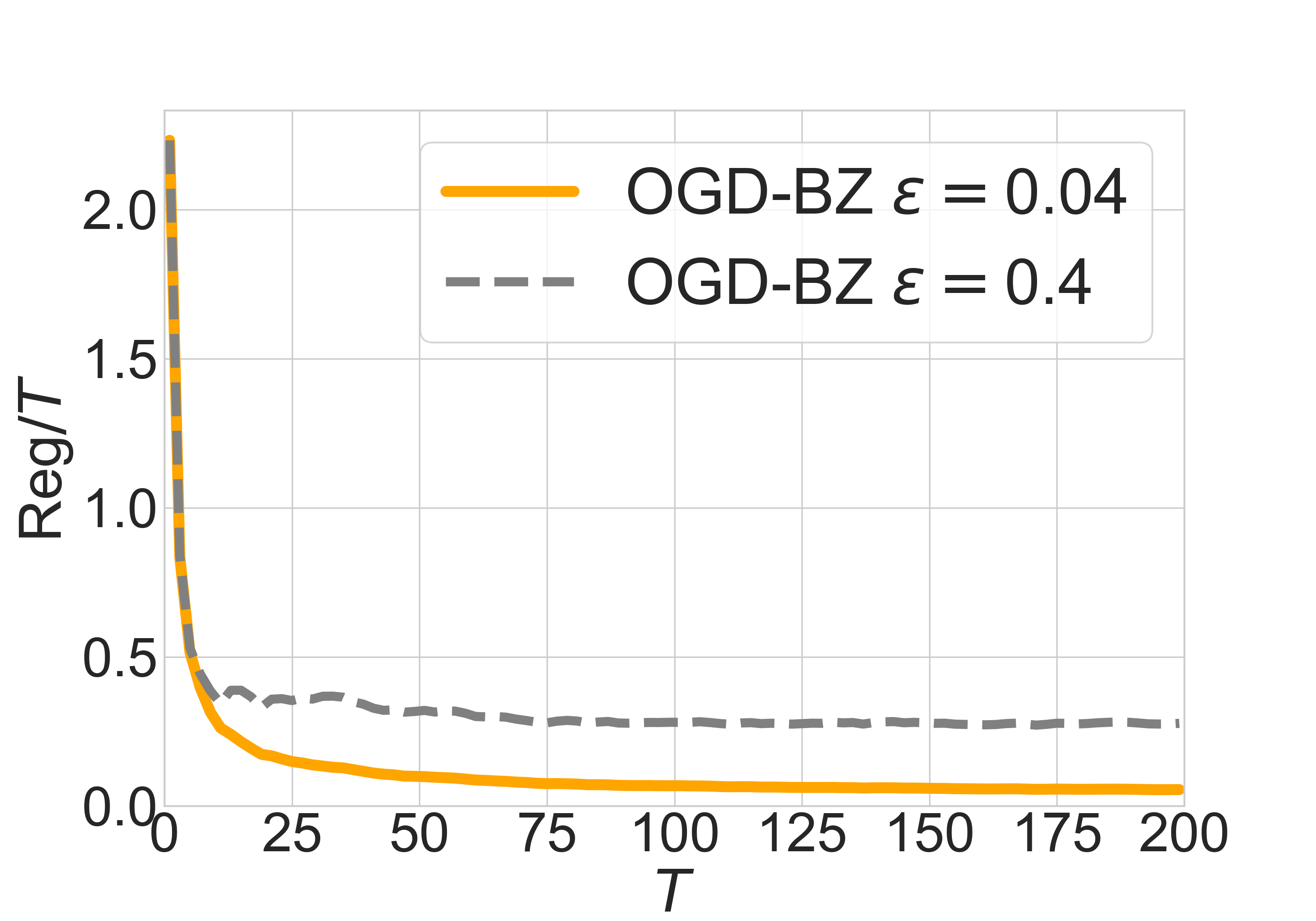}}\quad
    \subfigure[$x(t)$'s range]{\includegraphics[width=0.3\linewidth]{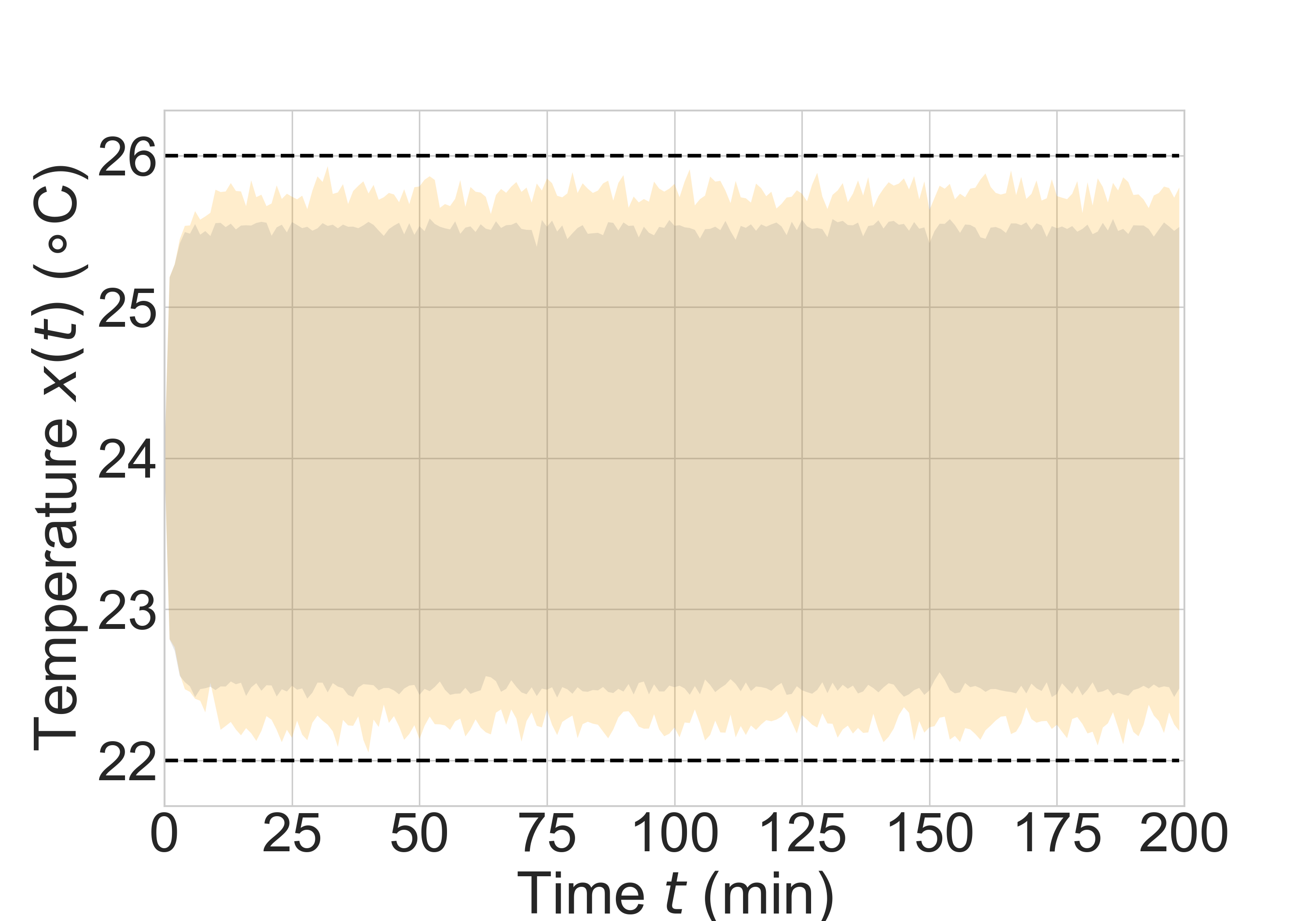}}\quad
    \subfigure[$u(t)$'s range]{\includegraphics[width=0.3\linewidth]{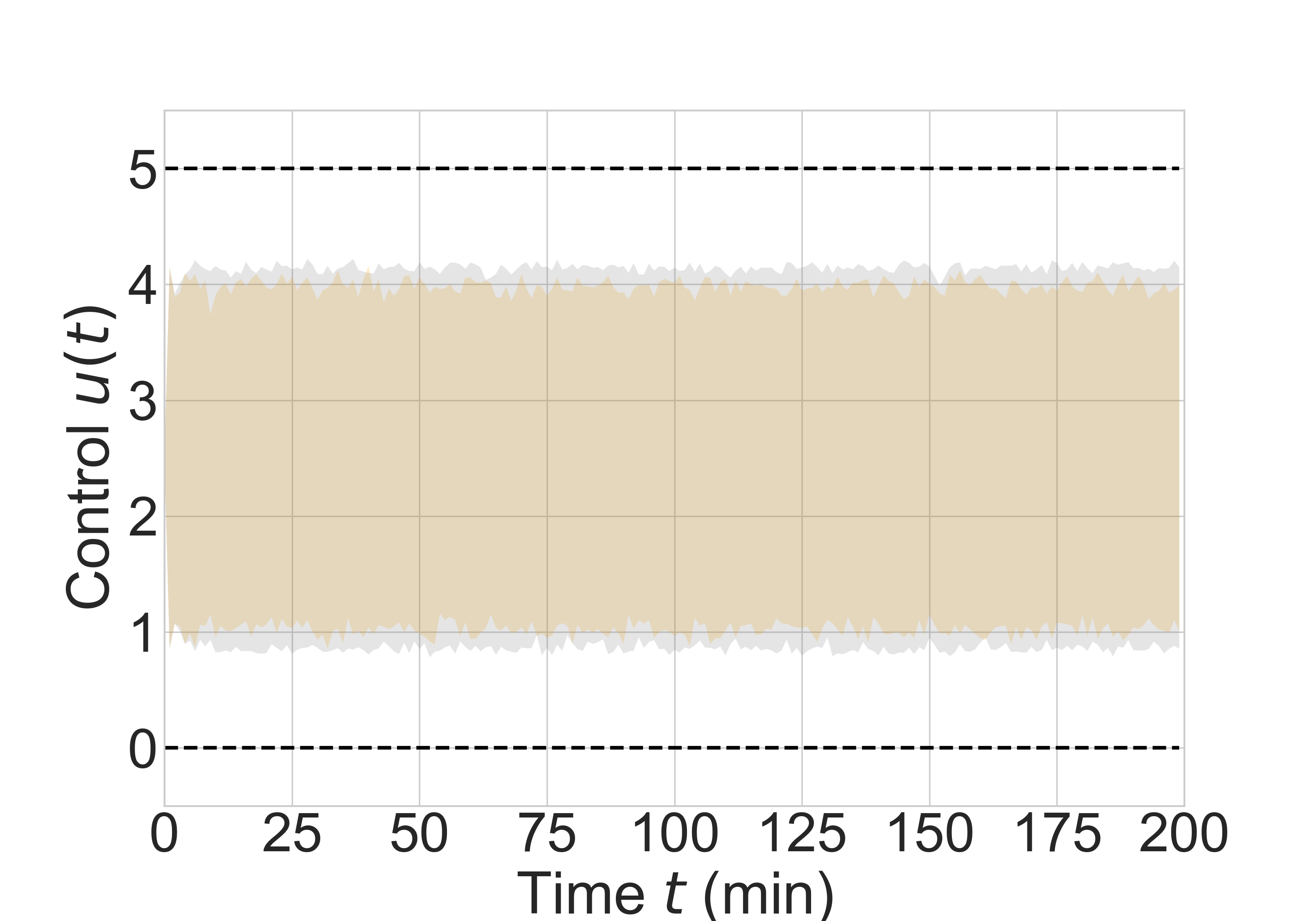}}
   
\caption{Comparison of OGD-BZ with buffer sizes $\epsilon=0.04$ and $\epsilon=0.4$. In Figure (b) and (c), the yellow shade represents the range of $x(t)$ generated by OGD-BZ with $\epsilon =0.04$, while the grey shade is generated by OGD-BZ with $\epsilon=0.4$.}
\label{fig:test}
\end{figure*}

Figure \ref{fig:test} plots the comparison of OGD-BZ with different buffer sizes. Specifically, $\epsilon=0.04$ is a properly chosen buffer size and $\epsilon=0.4$ offers larger buffer zones. From Figure \ref{fig:test}-(a), we can observe that the averaged regret with a properly chosen buffer size $\epsilon=0.04$ quickly diminishes to 0, which is consistent with Theorem \ref{thm: regret bdd}. In addition, Figure \ref{fig:test}-(b) and Figure \ref{fig:test}-(c) plot the range of $x(t)$ and $u(t)$ under random disturbances in 1000 trials to demonstrate the safety of OGD-BZ. With a larger buffer zone, i.e. $\epsilon=0.4$, the range of $x_t$ is smaller and further from the  boundaries, thus being safer. Interestingly, the range of $u(t)$ becomes slightly larger, which still satisfies the control constraints because the control constraints are not binding/active in this experiment and which  indicates more control power is used here to ensure a smaller range of $x(t)$ under disturbances. Finally, the regret with $\epsilon=0.4$ is  worse than that with $\epsilon=0.04$, which demonstrates the trade-off between safety and  performance and how the choices of the buffer size  affect this trade-off.



\section{Supplementary Proofs for Lemma \ref{lem: part iii first half}}

\subsection{Proof of Proposition \ref{prop: perturbation}}

Since $\Omega_2\subseteq \Omega_1$, we have $\min_{\Omega_2} f(x) - \min_{\Omega_1} f(x)\geq 0$. Let $x_1^*=\argmin_{\Omega_1} f(x)$. We will show that there exists $x_2^\dagger \in \Omega_2$ such that $\|x_1^*-x_2^\dagger\|_2\leq  \frac{\delta_1 \|\Delta\|_\infty}{\min_{i\in S} (h-C \mathring x)_i}$, where $S=\{i:\Delta_i>0\}$. Then, by the Lipschitz continuity, we can  prove the bound: $
\min_{\Omega_2} f(x) \!- \!\min_{\Omega_1} f(x) \leq f(x_2^\dagger)\!-\!f(x_1^*)\!\leq \!\frac{L\delta_1 \|\Delta\|_\infty}{\min_{i\in S} (h\!-\!C \mathring x)_i}.
$

In the following, we will show, more generally, that there exists $ x_2\in \Omega_2$ that is close to $x_1$ for any $x_1\in \Omega_1$. 
For ease of notation, we define $y=x-\mathring x$, 
$\Omega_1^y=\{y: Cy \leq h-C\mathring x\}$, and $\Omega_2^y=\{y: Cy \leq h-C\mathring x-\Delta\}$. Notice that $0\in \Omega_2^y$ and $(h-C\mathring x-\Delta)_i\geq 0$.
Besides, we have $y_1=x_1-\mathring x\in \Omega_1^y$. Further,  by the convexity of $\Omega_1^y$, we have $\lambda y_1\in \Omega_1^y$ for $0\leq \lambda\leq 1$. 

If $(C y_1)_i \leq (h-C \mathring x-\Delta)_i$ for all $i$, then $y_1 \in \Omega_2^y$ and $x_1\in \Omega_2$. So we can let $x_2=x_1$ and $\|x_2-x_1\|_2=0$.

If, instead, there exists a set $S'$ such that for any $i \in S'$, $(C y_1)_i > (h-C \mathring x-\Delta)_i$. Then, define $$\lambda=\min_{i\in S'}\frac{(h-C \mathring x-\Delta)_i}{(Cy_1)_i}.$$ Notice that $\lambda \in [0,1)$. 
We can show that $\lambda y_1 \in \Omega_2^y$ below. When $i \in S'$, $(\lambda C y_1)_i \leq (C y_1)_i\frac{ (h-C \mathring x-\Delta)_i}{(Cy_1)_i}=(h-C \mathring x-\Delta)_i$. When $i\not\in S'$, we have $(\lambda C y_1)_i \leq \lambda(h-C \mathring x-\Delta)_i\leq (h-C \mathring x-\Delta)_i$. Therefore, $\lambda y_1 \in \Omega_2^y$. Define $x_2=\lambda y_1 +\mathring x$, then $x_2 \in \Omega_2$.
Notice that
$\|x_1-x_2\|_2 =\|y_1-y_2\|_2=(1-\lambda) \|y_1\|_2
 \leq (1-\lambda) \delta_1
$.
Since $y_1 \in \Omega_1^y$, when $i \in S'$, we have $0\leq (h-C\mathring x-\Delta)_i<(Cy_1)_i\leq (h-C\mathring x)_i$. Therefore,  $\frac{(h-C \mathring x-\Delta)_i}{(Cy_1)_i}\geq \frac{(h-C \mathring x-\Delta)_i}{(h-C\mathring x)_i}=1-\frac{\Delta_i}{(h-C\mathring x)_i}$. Consequently, by $S'\subseteq S$, we have
$
1-\lambda \leq \max_{i\in S'}\frac{\Delta_i}{(h-C\mathring x)_i} \leq \frac{\|\Delta\|_\infty}{\min_{i\in S'}(h-C\mathring x)_i}
 \leq \frac{\|\Delta\|_\infty}{\min_{i\in S}(h-C\mathring x)_i}.
$

\subsection{Proof of Lemma \ref{lem: diameter and L bound}}
We first provide an explicit expression for $\Gamma_\epsilon$ and then prove the bounds on $\Gamma_\epsilon$  based on the explicit expression.

\begin{lemma}\label{lem: standard polytope form}
	For any $\epsilon\in \R$, $\bm M\in \Omega_\epsilon$ if and only if there exist
	$\{Y^x_{i,k,l}\}_{(1\leq i \leq k_x,  1\leq k \leq 2H,1\leq l\leq n)}$, $\{\!Y^u_{j,k,l}\!\}_{(1\leq j \leq k_u,1\leq k \leq 2H,1\leq l\leq n)} ,\{\!Z^{[i]}_{k,j}\!\}_{(1\leq i \leq H,1\leq k \leq m, 1\leq j\leq n)}$ such that
	\begin{align*}
	\begin{cases}
	&\sum_{k=1}^{2H}\sum_{l=1}^n Y^x_{i,k,l}\bar w \leq d_{x,i}-\epsilon,\  \forall 1\leq i\leq k_x\\
	& \sum_{k=1}^{2H}\sum_{l=1}^n Y^u_{j,k,l}\bar w \leq d_{u,j}-\epsilon,\  \forall 1\leq j \leq k_u\\
	& \sum_{j=1}^n Z^{[i]}_{k,j}\leq  2\sqrt n \kappa^3 (1-\gamma)^{i-1},  \forall 1\!\leq \!i \!\leq\! H,\! 1\!\leq \!k \!\leq\! m\\
	& -Y^x_{i,k,l}\leq (D_{x,i}^\top \mathring \Phi_k^x(\bm M))_l \leq Y^x_{i,k,l}, \quad \forall i,k,l\\
	&-Y^u_{j,k,l}\leq (D_{u,j}^\top \mathring \Phi_k^u(\bm M))_l \leq Y^u_{j,k,l}\quad \forall i,k,l\\
	&  -Z^{[i]}_{k,j}\leq M^{[i]}_{k,j}\leq Z^{[i]}_{k,j},\quad \forall i,k,j.
	\end{cases}
	\end{align*}
	Let $\vec W$ denote the vector containing the elements of $\bm M, \bm Y^x=\{Y^x_{i,k,l}\}, \bm Y^u=\{Y^u_{j,k,l}\}, \bm Z=\{Z^{[i]}_{k,j}\}$. Thus, the constraints  above can be written as
	$\Gamma_{\epsilon}=\{\vec W: C \vec W\leq h_\epsilon\}.$
	
\end{lemma}
Since Lemma \ref{lem: standard polytope form} holds for any $\epsilon\in \R$, we can similarly define $\Gamma_{-\epsilon_1-\epsilon_3}=\{\vec W: C \vec W\leq h_{-\epsilon_1-\epsilon_3}\}$ which is equivalent to $\Omega_{-\epsilon_1-\epsilon_3}$.  Lemma \ref{lem: standard polytope form} is based on a standard  reformulation method in constrained optimization to handle inequalities involving absolute values so the proof is omitted.

\textbf{Proof of (i).}
	Firstly, notice that 
	$
	\sum_{j=1}^n (M_{k,j}^{[i]})^2\leq 	\sum_{j=1}^n (Z_{k,j}^{[i]})^2 \leq 	(\sum_{j=1}^n Z_{k,j}^{[i]})^2 \leq 4n \kappa^6(1-\gamma)^{2i-2}
$.
	Then, 
	\begin{align*}
	\sum_{k=1}^m \sum_{i=1}^H 	\sum_{j=1}^n (M_{k,j}^{[i]})^2 \leq \sum_{k=1}^m \sum_{i=1}^H 	\sum_{j=1}^n (Z_{k,j}^{[i]})^2 \leq 4n m \kappa^6\frac{1}{2\gamma-\gamma^2}
	\end{align*}
	Similarly, by the first two  constraints in Lemma \ref{lem: standard polytope form} and by the definition of $\Gamma_{-\epsilon_1-\epsilon_3}$, we have
	$
	 \sum_{k=1}^{2H}\sum_{l=1}^n (Y^x_{i,k,l})^2\leq (d_{x,i}+\epsilon_1+\epsilon_3)^2/\bar w^2\leq (d_{x,i}+\epsilon_*)^2/\bar w^2$ and $
	 \sum_{k=1}^{2H}\sum_{l=1}^n (Y^u_{j,k,l})^2\leq (d_{u,j}+\epsilon_1+\epsilon_3)^2/\bar w^2\leq(d_{u,j}+\epsilon_*)^2/\bar w^2 
$.
	Therefore, 
	$
	   \sum_{i=1}^{k_x} \sum_{k=1}^{2H}\sum_{l=1}^n (Y^x_{i,k,l})^2\leq  \sum_{i=1}^{k_x}(d_{x,i}+\epsilon_*)^2/\bar w^2,$ and 
	$ \sum_{j=1}^{k_u} \sum_{k=1}^{2H}\sum_{l=1}^n (Y^u_{j,k,l})^2\leq \sum_{j=1}^{k_u}(d_{u,j}+\epsilon_*)^2/\bar w^2
$.
Consequently,
	\begin{align*}
	&\|\bm M\|_F^2 + \|\bm Y^x\|_F^2+ \|\bm Y^u \|_F^2 + \|\bm Z\|_F^2 \\
	\leq\ & \frac{8n m \kappa^6}{2\gamma-\gamma^2}+\frac{\sum_{i=1}^{k_x}(d_{x,i}+\epsilon_*)^2+\sum_{j=1}^{k_u}(d_{u,j}+\epsilon_*)^2}{\bar w^2}=\delta_1^2
	\end{align*}
	where $\delta_1=\Theta(\sqrt{mn}+\sqrt{k_c})$ by the boundedness of $\epsilon_*, d_x, d_u$. (Although $\delta_1$ depends linearly on $1/\bar w$, we will show $L=T G_f$ and $G_f$ is quadratic on $\bar w$ by Lemma \ref{lem: part ii}, hence, $L\delta_1$ is still linear with $\bar w$.)
	
	\textbf{Proof of (ii).} Since the gradient of $\mathring f_t(\bm M)$ is bounded by $G_f=\Theta(\sqrt{n^3 m H^3})$, the gradient of $\sum_{t=0}^T \mathring f_t(\bm M)$ is bounded by $LG_f=\Theta(T\sqrt{n^3 m H^3})$.
	
	\textbf{Proof of (iii).} Notice that the differences between $\Gamma_\epsilon$ and $\Gamma_{-\epsilon_1-\epsilon_3}$  come from the first two lines of the right-hand-side of inequalities in Lemma \ref{lem: standard polytope form}, which is $\epsilon+\epsilon_1+\epsilon_3$ in total. 
	
	\textbf{Proof of (iv).} From the proof of Theorem \ref{thm: feasibility}, we know that $\bm M(K_*)\in \Omega_{\epsilon_*-\epsilon_1-\epsilon_3}\subseteq \Omega_{\epsilon}$. Therefore, there exist corresponding $\bm Y^x(K_*), \bm Y^u(K_*), \bm Z(K_*)$ such that $\vec W^{\circ}=\text{vec}(\bm M(K_*),\bm Y^x(K_*), \bm Y^u(K_*), \bm Z(K_*))\in \Gamma_{\epsilon_*-\epsilon_1-\epsilon_3}\subseteq \Gamma_{\epsilon}$. Therefore,    $\min_{\{i:\Delta_i>0\}} (h_{-\epsilon_1-\epsilon_3}\!-\!C\vec W^{\circ})_i\geq \epsilon_1+\epsilon_3-(-\epsilon_*+\epsilon_1+\epsilon_3)=\epsilon_*$.


\section{Conclusion and Future Work}
This paper studies online optimal control with linear constraints  and linear dynamics with random disturbances. We propose OGD-BZ and  show that OGD-BZ can satisfy all the constraints despite  disturbances and ensure $\tilde O(\sqrt T)$ policy regret. 
There are many interesting future directions, e.g. (i) consider adversarial disturbances and robust stability, (ii) consider soft constraints and unbounded noises,  (iii) consider bandit feedback, (iv) reduce the regret bound's dependence on dimensions, (v) consider  unknown  systems, (vi) consider more general policies than    linear policies, (vii) prove logarithmic regrets for strongly convex costs, etc.

\section*{Acknowledgements}
This work was conducted while the first author was doing internship at
the MIT-IBM Watson AI Lab. We thank the helpful suggestions  from  Jeff Shamma, Andrea Simonetto, Yang Zheng, Runyu Zhang, and the reviewers.

\section*{Ethics Statement}
The primary motivation for this paper is to develop an online control algorithm under linear constraints on the states and actions, and under noisy linear dynamics. Some practical physical systems can be approximated by noisy linear dynamics and most practical systems have to satisfy certain constraints on the states and actions, such as data center cooling and robotics, etc.  Our proposed approach ensures to generate control policies that satisfy the constraints even under the uncertainty of unknown noises. Thus our algorithm can potentially be very beneficial for safety critical applications. However, note that our approach relies on a set of technical assumptions, as mentioned in the paper, which may not directly hold for all practical applications. Hence, when applying our algorithm, particular cares are needed when modeling the system and the constraints.

\bibliography{citation4onlineLQRconstraint}


\input{appendix_technical}

\end{document}

%% file: appendix_technical.tex
\newpage
\onecolumn
\section*{Appendix}
In the following, we provide a complete proof for the technical results in the main submission. Specifically, Appendix A provides some helping lemmas, Appendix B provides proofs for Section 5.1, and Appendix C provides proofs for Section 5.2.

\section{Helping lemmas}
In this section, we provide some  technical lemmas that will be useful in the proofs. These lemmas are similar to the results established in \cite{agarwal2019online,agarwal2019logarithmic} but involve slightly different coefficients in the bounds because we define $\M$ differently.

\paragraph{- A property of $(\kappa,\gamma)$  strongly stable matrices.}
\begin{lemma}\label{lem: |AK^k|2 bdd}
	When $K$ is $(\kappa,\gamma)$  strongly stable, then $\|(A-BK)^k\|_2 \leq \kappa^2(1-\gamma)^k$ for any integer $k\geq 0$.
\end{lemma}
\begin{proof}
	By Definition \ref{def: diagonally strongly stable}, there exist $H$ and $L$ such that $A-BK=H^{-1}L H$. Thus, $(A-BK)^k=H^{-1}L^k H$, and $\|(A-BK)^k\|_2 \leq \|H^{-1}\|_2\|L^k\|_2\|H\|_2\leq  \|H^{-1}\|_2\|L\|_2^k\|H\|_2\leq \kappa^2(1-\gamma)^k$ when $k\geq 1$. When $k=0$, $\|(A-BK)^0\|_2=\|I\|_2=1\leq \kappa^2$ since $\kappa\geq 1$.
\end{proof}

\paragraph{- A property of set $\M$.}

\begin{lemma}\label{lem: M's L2 norm}

	For $\bm M\in \M=\{ \bm M=\{ M^{[1]}, \dots, M^{[H]} \}: \ \|M^{[i]}\|_\infty \leq 2\sqrt n \kappa^3 (1-\gamma)^{i-1}\}$,  we have
	$$\|M^{[i]}\|_2 \leq2\sqrt{mn}\kappa^3 (1-\gamma)^{i-1}$$
\end{lemma}
\begin{proof}
	The proof is by $\|M^{[i]}\|_2 \leq \sqrt m \|M^{[i]}\|_\infty$ due to properties of matrix norms.
\end{proof}

\paragraph{- Bounds on the states $x_t$, the actions $u_t$, the approximate states $ \tilde x_t$, and the approximate actions $\tilde u_t$.}

\begin{lemma}\label{lem: bdd on Phix,u}
	For $\bm M_0, \dots, \bm M_{T}\in \M$,  we have
	\begin{align*}
	&	\|\Phi^x_{k}(\bm M_{t-H:t-1})\|_2\leq \kappa^2(1-\gamma)^{k-1}\one_{(k\leq H)}+\phi H (1-\gamma)^{k-2}\one_{(k\geq 2)}\\
	&	\|\Phi^u_{k}(\bm M_{t-H:t})\|_2\leq \kappa^3(1-\gamma)^{k-1}\one_{(k\leq H)}(2\sqrt{nm}+1)+ \kappa\phi H(1-\gamma)^{k-2}\one_{(k\geq 2)}
	\end{align*}
	where $\phi=2\kappa^5 \kappa_B \sqrt{mn} $.

\end{lemma}
\begin{proof}
	By Proposition \ref{prop: approx}, Lemma \ref{lem: |AK^k|2 bdd}, the definition of $\M$, and Definition \ref{def: disturbance action}, we have
	\begin{align*}
	\|	\Phi_{k}^x(\bm M_{t-H:t-1})\|_2&=	\|\AK^{k-1}\one_{(k\leq H)}+ \sum_{i=1}^H \AK^{i-1}B M_{t-i}^{[k-i]}\one_{(1\leq k-i\leq H)}\|\\
	& \leq \| \AK^{k-1}\|_2\one_{(k\leq H)}+ \sum_{i=1}^H  \|\AK^{i-1}\|_2\|B\|_2 \|M_{t-i}^{[k-i]}\|_2\one_{(1\leq k-i\leq H)}\\
	& \leq \kappa^2(1-\gamma)^{k-1}\one_{(k\leq H)}+ \sum_{i=1}^H \kappa^2(1-\gamma)^{i-1}\kappa_B\sqrt m\|M_{t-i}^{[k-i]}\|_\infty \one_{(1\leq k-i\leq H)}\\
	& \leq \kappa^2(1-\gamma)^{k-1}\one_{(k\leq H)}+ 2\kappa^5 \kappa_B \sqrt{mn}  \sum_{i=1}^H (1-\gamma)^{i-1}(1-\gamma)^{k-i-1}\one_{(1\leq k-i\leq H)}\\
	& \leq \kappa^2(1-\gamma)^{k-1}\one_{(k\leq H)}+2\kappa^5 \kappa_B \sqrt{mn}  H (1-\gamma)^{k-2}\one_{(k\geq 2)}\\
	&=\kappa^2(1-\gamma)^{k-1}\one_{(k\leq H)}+ \phi H (1-\gamma)^{k-2}\one_{(k\geq 2)},
	\end{align*}
	where $\phi=2\kappa^5 \kappa_B \sqrt{mn}  $. 
	
	Further, by Proposition \ref{prop: approx}, Lemma \ref{lem: |AK^k|2 bdd}, the definition of $\M$,  Definition \ref{def: disturbance action}, and the bound above, we have
	\begin{align*}
	\|	\Phi_{k}^u(\bm M_{t-H:t})\|_2&=\|	M_t^{[k]}\one_{(k\leq H)}-\Kb\Phi_{k}^x(\bm M_{t-H:t-1})\|_2\\
	& \leq \|M_t^{[k]}\|_2\one_{(k\leq H)}+\|\Kb\|_2\|\Phi_{k}^x(\bm M_{t-H:t-1})\|_2\\
	& \leq \sqrt m\|M_t^{[k]}\|_\infty\one_{(k\leq H)}+\|\Kb\|_2\|\Phi_{k}^x(\bm M_{t-H:t-1})\|_2\\
	& \leq 2\sqrt{nm}\kappa^3 (1-\gamma)^{k-1}\one_{(k\leq H)}+\kappa\left(\kappa^2(1-\gamma)^{k-1}\one_{(k\leq H)}+ \phi H (1-\gamma)^{k-2}\one_{(k\geq 2)}\right)\\
	&=\kappa^3(1-\gamma)^{k-1}\one_{(k\leq H)}(2\sqrt{nm}+1)+ \kappa\phi H(1-\gamma)^{k-2}\one_{(k\geq 2)}
	\end{align*}

\end{proof}

\begin{lemma}\label{lem: bdd on xt ut xt tilde ut tilde}
	
	When implementing  $\bm M_0, \dots, \bm M_{T}\in \M$, we have 
	\begin{align*}
	\max(	\|x_t\|_2, \|\tilde x_2\|_2)\leq  b_x,\quad \max(	\|u_t\|_2, \|\tilde u_2\|_2)\leq b_u
	\end{align*}
	where 	 $b_x\coloneqq \frac{\sqrt{n}\bar w (\kappa^2+\phi H)}{(1-\kappa^2(1-\gamma)^H)\gamma}$ and $b_u= \kappa  b_x+2\sqrt{mn}\kappa^3\bar w /\gamma$. 
	Define $b=\max(b_x, b_u)=\frac{\kappa\sqrt{n}\bar w (\kappa^2+\phi H)}{(1-\kappa^2(1-\gamma)^H)\gamma}+2\sqrt{mn}\kappa^3\bar w /\gamma$.
	
	Consequently, when $H\geq \frac{\log(2\kappa^2)}{\log((1-\gamma)^{-1})}$, we have $b \leq 8 \sqrt{mn^2} H \bar w \kappa^6 \kappa_B/\gamma$.
\end{lemma}

\begin{proof}
	Firstly, we bound $\|\tilde x_t\|_2$ below.
	\begin{align*}
	\|\tilde x_t \|_2 &=\| \sum_{k=1}^{2H}\Phi^x_{k}(\bm M_{t-H:t-1})w_{t-k}\|_2 \leq \sum_{k=1}^{2H} \|\Phi^x_{k}(\bm M_{t-H:t-1})\|_2 \sqrt{n}\bar w\\
	& \leq \sqrt{n}\bar w\sum_{k=1}^{2H}\left(\kappa^2(1-\gamma)^{k-1}\one_{(k\leq H)}+\phi H (1-\gamma)^{k-2}\one_{(k\geq 2)}\right)\\
	& \leq \sqrt{n}\bar w ( \kappa^2/\gamma+ \phi H/\gamma)\leq  b_x
	\end{align*}
	Next, we bound $\|x_t\|_2$ by $\|\tilde x_t\|_2$'s bound and induction based on the recursive equation in Proposition \ref{prop: approx}. Specifically, we first note that $\|x_{t-H}\|_2=0\leq b_x$ when $t=0$.  At any $t\geq 0$, suppose $\|x_{t-H}\|_2\leq b_x$, then 
	\begin{align*}
	\|x_t\|_2&=\|\AK^H x_{t-H}+ \tilde x_t\|_2\leq \kappa^2(1-\gamma)^H b_x +  \sqrt{n}\bar w ( \kappa^2/\gamma+ \phi H/\gamma)= b_x
	\end{align*}
	Consequently, we have proved that $\|x_t\|_2 \leq b_x$ for all $t$.
	
	Similarly, we bound $\|\tilde u_t\|_2$ and $\|u_t\|_2$ below. 
	\begin{align*}
	\|\tilde u_t\|_2&=\left\|-\Kb \tilde x_t + \sum_{i=1}^H M_t^{[i]}w_{t-i}\right\|_2\leq \|-\Kb \tilde x_t\|_2 + \sum_{i=1}^H \|M_t^{[i]}w_{t-i}\|_2\\
	& \leq \|-\Kb \tilde x_t\|_2 + \sum_{i=1}^H\sqrt m \|M_t^{[i]}w_{t-i}\|_\infty\\
	& \leq \kappa b_x+ \sum_{i=1}^H\sqrt m (2\sqrt n)\kappa^3(1-\gamma)^{i-1}\bar w= \kappa b_x +2\sqrt{mn} \kappa^3\bar w /\gamma = b_u
	\end{align*}

	Further, by Proposition \ref{prop: approx}, we have 
	\begin{align*}
	\|u_t\|_2&=\left\|-\Kb  x_t + \sum_{i=1}^H M_t^{[i]}w_{t-i}\right\|_2\leq \|-\Kb  x_t \|+\sum_{i=1}^H \sqrt m \| M_t^{[i]}w_{t-i}\|_\infty \\
	&\leq \kappa  b_x+ \sum_{i=1}^H 2\sqrt{mn}\kappa^3 (1-\gamma)^{i-1}\bar w \\
	& \leq \kappa  b_x+2\sqrt{mn}\kappa^3\bar w /\gamma=b_u
	\end{align*}
	
	Finally, notice that when $H\geq \frac{\log(2\kappa^2)}{\log((1-\gamma)^{-1})}$, we have $\kappa^2(1-\gamma)^{H} \leq 1/2$, and thus $b_x \leq 2\sqrt n \bar w(\kappa^2+\phi H)/\gamma$ and $b$'s bound follows naturally. 
	
\end{proof}

\paragraph{- Properties of $f_t(\bm M_{t-H:t})$ and $\mathring f_t(\bm M_t)$}
\begin{lemma}\label{lem: ft(Mt) property} 
	Consider any $\bm M_t \in \M$ and any $\tilde{\bm M_t}\in \M$ for all $t$, then 
			$$ |f_t(\bm M_{t-H:t})-f_t(\tilde{\bm M}_{t-H:t})| \leq  Gb(1+\kappa)\sqrt n \bar w\kappa^2\kappa_B \sum_{i=0}^H(1-\gamma)^{\max(i-1,0)} \sum_{j=1}^H\| M^{[j]}_{t-i}-\tilde{ M}^{[j]}_{t-i}\|_2.$$
			
			Further, $\|\nabla \mathring f_t(\bm M_t)\|_F\leq G_f$ for $\bm M_t \in \M$, where $G_f=\Theta(Gb(1+\kappa)\sqrt n \bar w\kappa^2\kappa_B \sqrt H \frac{1+\gamma}{\gamma})$.
			
			Consequently, when $H\geq \frac{\log(2\kappa^2)}{\log((1-\gamma)^{-1})}$, then $G_f\leq \Theta(\sqrt{n^3 H^3m})$.

\end{lemma}
\begin{proof}

 Let $\tilde x_t$ and $\dbtilde x_t$ denote the approximate states generated by $\bm M_{t-H:t-1}$ and $\tilde{\bm M}_{t-H:t-1}$ respectively. Define $\tilde u_t$ and $\dbtilde u_t$ similarly. We have
	\begin{align*}
	\|\tilde x_t-\dbtilde x_t\|_2& =\|\sum_{k=1}^{2H} (\Phi_{k}^x(\bm M_{t-H:t-1})-\Phi_{k}^x(\tilde{\bm M}_{t-H:t-1}))w_{t-k}\|_2\\
	& \leq \sum_{k=1}^{2H}\|\Phi_{k}^x(\bm M_{t-H:t-1})-\Phi_{k}^x(\tilde{\bm M}_{t-H:t-1})\|_2\sqrt n \bar w\\
	& \leq \sqrt n \bar w\sum_{k=1}^{2H}\| \sum_{i=1}^H \AK^{i-1}B (M_{t-i}^{[k-i]}-\tilde M_{t-i}^{[k-i]})\one_{(1\leq k-i \leq H)}\|_2\\
	&\leq \sqrt n \bar w\sum_{k=1}^{2H}\sum_{i=1}^H \kappa^2 (1-\gamma)^{i-1}\kappa_B\|M_{t-i}^{[k-i]}-\tilde M_{t-i}^{[k-i]}\|_2\one_{(1\leq k-i \leq H)}\\
	& = \sqrt n \bar w\kappa^2\kappa_B \sum_{i=1}^H(1-\gamma)^{i-1} \sum_{j=1}^H \|M_{t-i}^{[j]}-\tilde M_{t-i}^{[j]}\|_2
	\end{align*}
	and 
	\begin{align*}
	\|\tilde u_t-\dbtilde u_t\|_2&=\|-\Kb \tilde x_t +\Kb \dbtilde x_t+ \sum_{i=1}^H M_t^{[i]}w_{t-i}- \sum_{i=1}^H \tilde M_t^{[i]}w_{t-i}\|_2 \\
	&\leq \kappa 	\|\tilde x_t-\dbtilde x_t\|_2 + \sqrt n  \bar w \sum_{i=1}^H \|M_t^{[i]}-\tilde M_t^{[i]}\|_2\\
	& \leq \sqrt n \bar w\kappa^3\kappa_B \sum_{i=0}^H(1-\gamma)^{\max(i-1,0)} \sum_{j=1}^H\| M^{[j]}_{t-i}-\tilde{ M}^{[j]}_{t-i}\|_2 
	\end{align*}
	
	Consequently, by Assumption \ref{ass: bounded Hessian largest evalue} and Lemma \ref{lem: bdd on xt ut xt tilde ut tilde}, we can prove the first bound in the lemma's statement below.
	\begin{align*}
	|f_t(\bm M_{t-H:t-1})-f_t(\tilde{\bm M}_{t-H:t-1})|& = |c_t(\tilde x_t, \tilde u_t)-c_t(\dbtilde x_t, \dbtilde u_t)|\\
	& \leq Gb (\|\tilde x_t-\dbtilde x_t\|_2+\|\tilde u_t-\dbtilde u_t\|_2)\\
	&\leq Gb(1+\kappa)\sqrt n \bar w\kappa^2\kappa_B \sum_{i=0}^H(1-\gamma)^{\max(i-1,0)} \sum_{j=1}^H\| M^{[j]}_{t-i}-\tilde{ M}^{[j]}_{t-i}\|_2.
	\end{align*}
	
	Next, we prove the gradient bound on $\mathring f_t(\bm M_t)$. Define a set $\M_{out,H}=\{\bm M: \|M[k]\|_\infty \leq 4\kappa^2 \sqrt n (1-\gamma)^{k-1}\}$, whose interior contains $\M_H$. Similar to Lemma \ref{lem: ft(Mt) property}, we can show $\mathring f_t(\bm M+\Delta \bm M)-\mathring f_t(\bm M)\leq \Theta(b\sqrt n w_{\max} \sum_{i=0}^H(1-\gamma)^{\max(i-1,0)} \sum_{j=1}^H\|\Delta M{[j]}\|_2)$ for any $\bm M\in \M_H$ and $\bm M+\Delta \bm M+ \M_{out,H}$. 
	
	By the definition of the operator's norm, we have
	\begin{align*}
	&	\|\nabla \mathring f_t(\bm M)\|_F= \sup_{\Delta \bm M\not =0,\bm M+\Delta \bm M+ \M_{out,H}} \frac{\langle \nabla \mathring f_t(\bm M), \Delta \bm M\rangle}{\|\Delta \bm M\|_F}\\
		& \leq \sup_{\Delta \bm M\not =0,\bm M+\Delta \bm M+ \M_{out,H}} \frac{\mathring f_t(\bm M+\Delta \bm M)-\mathring f_t(\bm M)}{\|\Delta \bm M\|_F}\\
		& \leq \sup_{\Delta \bm M\not =0,\bm M+\Delta \bm M+ \M_{out,H}}\frac{\Theta(b\sqrt n w_{\max} \sum_{i=0}^H(1-\gamma)^{\max(i-1,0)} \sum_{j=1}^H\|\Delta M^{[j]}\|_2)}{\|\Delta \bm M\|_F}\\
		& \leq \sup_{\Delta \bm M\not =0,\bm M+\Delta \bm M+ \M_{out,H}}\frac{\Theta(b\sqrt n w_{\max} \sum_{i=0}^H(1-\gamma)^{\max(i-1,0)} \sqrt H\|\Delta \bm M\|_F)}{\|\Delta \bm M\|_F}\\
		& \leq \Theta(b\sqrt n w_{\max} \sqrt H \frac{1+\gamma}{\gamma})
	\end{align*}
\end{proof}

\section{Proofs of the lemmas used in the proof of Theorem \ref{thm: feasibility}}
\subsection{Proof of Lemma \ref{lem: bdd xt-H part}}
The proof is straightforward from Lemma \ref{lem: |AK^k|2 bdd} and
	 Lemma \ref{lem: bdd on xt ut xt tilde ut tilde}.
	\begin{align*}
	\|D_x \AK^H x_{t-H}\|_\infty &\leq \|D_x\|_\infty \|\AK^H x_{t-H}\|_\infty \leq 	\|D_x\|_\infty \| \AK^H x_{t-H}\|_2 \\
	&\leq \|D_x\|_\infty \| \AK^H\|_2 \|x_{t-H}\|_2 \leq \|D_x\|_\infty \kappa^2(1-\gamma)^H b\leq \epsilon_1(H)\\
	\|D_u\Kb \AK^H x_{t-H}\|_\infty &\leq \|D_u\|_\infty \|\Kb \AK^H x_{t-H}\|_\infty\leq 	\|D_u\|_\infty \|\Kb \AK^H x_{t-H}\|_2 \\
	&\leq \|D_u\|_\infty \|\Kb\|_2\| \AK^H\|_2 \|x_{t-H}\|_2 \leq \|D_u\|_\infty \kappa^3(1-\gamma)^H b\leq\epsilon_1(H) 
	\end{align*}
	Hence, we have $\epsilon_1(H)=c_1 n\sqrt{m} H(1-\gamma)^H$ and $c_1=8 \bar w \kappa^{9}\kappa_B\max(\|D_x\|_\infty, \|D_u\|_\infty)/\gamma$.

\subsection{Proof of Lemma \ref{lem: g(Mt:t-H) to g(Mt) bdd error}}

Lemma \ref{lem: g(Mt:t-H) to g(Mt) bdd error} is proved by first establishing a smoothness property of $g_i^x(\cdot)$ and $g_j^u(\cdot)$ in Lemma \ref{lem: g is Lg Lip cont} and then leveraging the slow updates of OGD. The  details of the proof are provided below.

\begin{lemma}\label{lem: g is Lg Lip cont}

	Consider any $\bm M_t \in \M$ and any $\tilde{\bm M_t}\in \M$ for all $t$, then 
	\begin{align*}
	&	\max_{1\leq i \leq k_x}\left| g^x_{i}(\bm M_{t-H:t-1})-g^x_{i}(\tilde{\bm M}_{t-H:t-1})\right| \leq L_g(H) \sum_{k=1}^H (1-\gamma)^{k-1}\|\bm M_{t-k}-\tilde{\bm M}_{t-k}\|_F\\
	&\max_{	1\leq j \leq k_u}\left| g^u_{j}(\bm M_{t-H:t})-g^u_{j}(\tilde{\bm M}_{t-H:t})\right| \leq L_g(H) \sum_{k=0}^H (1-\gamma)^{\max(k-1,0)}\|\bm M_{t-k}-\tilde{\bm M}_{t-k}\|_F
	\end{align*}
	where $L_g(H)=\bar w \sqrt n\max(\|D_x\|_\infty, \|D_u\|_\infty )\kappa^3\kappa_B \sqrt H.$
\end{lemma}

\begin{proof}
We first provide a bound on $\|D_{x,i}^\top \Phi^x_{k}(\bm M_{t-H:t-1})\|_1-\|D_{x,i}^\top \Phi^x_{k}(\tilde{\bm M}_{t-H:t-1})\|_1$.
	\begin{align*}
	&|	\|D_{x,i}^\top \Phi^x_{k}(\bm M_{t-H:t-1})\|_1-\|D_{x,i}^\top \Phi^x_{k}(\tilde{\bm M}_{t-H:t-1})\|_1| \leq \| D_{x,i}^\top \Phi^x_{k}(\bm M_{t-H:t-1})-D_{x,i}^\top \Phi^x_{k}(\tilde{\bm M}_{t-H:t-1})\|_1\\
	=\ & \left\|D_{x,i}^\top\left(\sum_{s=1}^H \AK^{s-1}B (M_{t-s}^{[k-s]}-\tilde M_{t-s}^{[k-s]}) \one_{(1\leq k-s\leq H)} \right)\right\|_1\\
	= \ & \left\|\sum_{s=1}^H D_{x,i}^\top\AK^{s-1}B (M_{t-s}^{[k-s]}-\tilde M_{t-s}^{[k-s]}) \one_{(1\leq k-s\leq H)}\right\|_1\\
	\leq \ & \sqrt n \left\|\sum_{s=1}^H D_{x,i}^\top\AK^{s-1}B (M_{t-s}^{[k-s]}-\tilde M_{t-s}^{[k-s]}) \one_{(1\leq k-s\leq H)}\right\|_2\\
	\leq \ &\sqrt n \sum_{s=1}^H \| D_{x,i}^\top\|_2\|\AK^{s-1}\|_2\|B\|_2 \|M_{t-s}^{[k-s]}-\tilde M_{t-s}^{[k-s]}\|_2 \one_{(1\leq k-s\leq H)}\\
	\leq \ & \sqrt n \sum_{s=1}^H  \| D_{x,i}^\top\|_1 \kappa^2 (1-\gamma)^{s-1}\kappa_B \|M_{t-s}^{[k-s]}-\tilde M_{t-s}^{[k-s]}\|_2 \one_{(1\leq k-s\leq H)}\\
	\leq \ & \sqrt n \|D_x\|_\infty \kappa^2 \kappa_B \sum_{s=1}^H (1-\gamma)^{s-1} \|M_{t-s}^{[k-s]}-\tilde M_{t-s}^{[k-s]}\|_2\one_{(1\leq k-s\leq H)}
	\end{align*}
	Therefore, for any $1\leq i \leq k_x$, we have
	\begin{align*}
	&\quad \ 	| g^x_{i}(\bm M_{t-H:t-1})-g^x_{i}(\tilde{\bm M}_{t-H:t-1})| \\
	&\leq \bar w \sum_{k=1}^{2H}|	\|D_{x,i}^\top \Phi^x_{k}(\bm M_{t-H:t-1})\|_1-\|D_{x,i}^\top \Phi^x_{k}(\tilde{\bm M}_{t-H:t-1})\|_1| \\
	&\leq \bar w  \sqrt n \|D_x\|_\infty \kappa^2 \kappa_B\sum_{k=1}^{2H} \sum_{s=1}^H  (1-\gamma)^{s-1} \|M_{t-s}^{[k-s]}-\tilde M_{t-s}^{[k-s]}\|_2\one_{(1\leq k-s\leq H)}\\
	& = \bar w  \sqrt n \|D_x\|_\infty \kappa^2 \kappa_B \sum_{s=1}^H \sum_{k=1}^{2H} (1-\gamma)^{s-1}\|M_{t-s}^{[k-s]}-\tilde M_{t-s}^{[k-s]}\|_2\one_{(1\leq k-s\leq H)}\\
	& \leq \bar w  \sqrt n \|D_x\|_\infty \kappa^2 \kappa_B \sum_{s=1}^H \sum_{k=1}^{2H} (1-\gamma)^{s-1}\|M_{t-s}^{[k-s]}-\tilde M_{t-s}^{[k-s]}\|_F\one_{(1\leq k-s\leq H)}\\
	&\leq \bar w  \sqrt n \|D_x\|_\infty \kappa^2 \kappa_B \sqrt H \sum_{s=1}^H(1-\gamma)^{s-1} \|\bm M_{t-s}-\tilde{\bm M}_{t-s}\|_F
	\end{align*}

	Next, we provide a bound on $\|D_{u,j}^\top \Phi^u_{k}(\bm M_{t-H:t})\|_1-\|D_{u,j}^\top \Phi^u_{k}(\tilde{\bm M}_{t-H:t})\|_1$.
	\begin{align*}
	&|	\|D_{u,j}^\top \Phi^u_{k}(\bm M_{t-H:t})\|_1-\|D_{u,j}^\top \Phi^u_{k}(\tilde{\bm M}_{t-H:t})\|_1| \leq \| D_{u,j}^\top (\Phi^u_{k}(\bm M_{t-H:t})-\Phi^u_{k}(\tilde{\bm M}_{t-H:t})\|_1\\
	\leq \ & \|D_{u,j}^\top (M_t^{[k]}-\tilde M_t^{[k]})\|_1\one_{(k\leq H)}+ \left\| \sum_{s=1}^H D_{u,j}^\top\Kb\AK^{s-1}B (M_{t-s}^{[k-s]}-\tilde M_{t-s}^{[k-s]})\right\|_1 \one_{(1\leq k-s\leq H)}\\
	\leq \ & \sqrt n  \|D_{u,j}^\top (M_t^{[k]}-\tilde M_t^{[k]})\|_2\one_{(k\leq H)}+ \sqrt n \left\| \sum_{s=1}^H D_{u,j}^\top\Kb\AK^{s-1}B (M_{t-s}^{[k-s]}-\tilde M_{t-s}^{[k-s]})\right\|_2 \one_{(1\leq k-s\leq H)}\\
	\leq \ & \sqrt n \|D_{u,j}^\top \|_2 \| M_t^{[k]}-\tilde M_t^{[k]}\|_2\one_{(k\leq H)}+ \sqrt n \sum_{s=1}^H \|D_{u,j}\|_2 \kappa^3(1-\gamma)^{s-1}\kappa_B \| M_{t-s}^{[k-s]}-\tilde M_{t-s}^{[k-s]}\|_2 \one_{(1\leq k-s\leq H)}\\
	\leq \ & \sqrt n \|D_{u,j}^\top \|_1 \| M_t^{[k]}-\tilde M_t^{[k]}\|_2\one_{(k\leq H)}+ \sqrt n \sum_{s=1}^H \|D_{u,j}\|_1 \kappa^3(1-\gamma)^{s-1}\kappa_B \| M_{t-s}^{[k-s]}-\tilde M_{t-s}^{[k-s]}\|_2 \one_{(1\leq k-s\leq H)}
	\end{align*}
	Therefore, for any $1\leq j\leq k_u$, we have
	\begin{align*}
	&\quad\ | g^u_{j}(\bm M_{t-H:t})-g^u_{j}(\tilde{\bm M}_{t-H:t})| \leq \sum_{k=1}^{2H} \bar w | 	\|D_{u,j}^\top \Phi^u_{k}(\bm M_{t-H:t})\|_1-\|D_{u,j}^\top \Phi^u_{k}(\tilde{\bm M}_{t-H:t})\|_1|\\
	&\leq \bar w \sqrt n \|D_u\|_\infty  \sqrt H \|\bm M_t-\tilde{\bm M}_t\|_F + \bar w \sqrt n \|D_u\|_\infty \kappa^3\kappa_B \sqrt H \sum_{s=1}^H (1-\gamma)^{s-1}\|\bm M_{t-s}-\tilde{\bm M}_{t-s}\|_F \\
	& \leq  \bar w \sqrt n \|D_u\|_\infty \kappa^3\kappa_B \sqrt H \sum_{s=0}^H(1-\gamma)^{\max(s-1,0)} \|\bm M_{t-s}-\tilde{\bm M}_{t-s}\|_F
	\end{align*}
	where the last inequality uses $\kappa\geq 1, \kappa_B\geq 1$. 	
\end{proof}

\begin{proof}[Proof of Lemma \ref{lem: g(Mt:t-H) to g(Mt) bdd error}]
Firstly, by OGD's definition and Lemma \ref{lem: ft(Mt) property}, we have $\|\bm M_t-\bm M_{t-1}\|_F\leq \eta G_f$ and  $\|\bm M_t-\bm M_{t-k}\|_F\leq k\eta G_f$.   By Lemma \ref{lem: g is Lg Lip cont}, we have
	\begin{align*}
	&	\max_{1\leq i \leq k_x}\left| g^x_{i}(\bm M_{t-H+1:t})-g^x_{i}({\bm M}_{t}, \dots, \bm M_t)\right| \leq L_g(H) \sum_{k=1}^{H-1} (1-\gamma)^{k-1}\|\bm M_{t-k}-\bm M_t\|_F\\
	& \leq L_g(H) \sum_{k=1}^{H-1} (1-\gamma)^{k-1}k \eta G_f \leq L_g(H)\eta G_f \frac{1}{\gamma^2}
	\end{align*}
	and 
	\begin{align*}
	&	\max_{1\leq j \leq k_u}\left| g^u_{j}(\bm M_{t-H:t})-g^u_{j}(\tilde{\bm M}_{t}, \dots, \bm M_t)\right| \leq L_g(H) \sum_{k=1}^H (1-\gamma)^{k-1}\|\bm M_{t-k}-\bm M_t\|_F\leq \epsilon_2(\eta, H)\\
	& \leq L_g(H) \sum_{k=1}^H (1-\gamma)^{k-1}k \eta G_f \leq L_g(H)\eta G_f \frac{1}{\gamma^2}\leq \epsilon_2(\eta, H)
	\end{align*}
	where $\epsilon_2(\eta, H)=c_2 n^2 \sqrt m H^2\eta $. 
\end{proof}

\subsection{Proof of Lemma \ref{lem: define M_ap and epsilon3}}

Firstly, notice that when implementing $u_t^K=-K x_t^K$ and $x_0=0$, we have
\begin{align*}
x_t^K&=\sum_{s=1}^t A_{K}^{s-1}w_{t-s}, \qquad u_t^K=-\sum_{s=1}^t K A_{K}^{s-1}w_{t-s},
\end{align*}
where $A_K=A-BK$.

In addition, when implementing a disturbance-action controller $\bm M$, the state satisfies 
\begin{align*}
x_t^{\bm M(K)}&=\sum_{s=1}^t \tilde \Phi^x_{t,s}(\bm M(K))w_{t-s}, \quad \text{ where }
\tilde \Phi^x_{s}(\bm M(K))= \AK^{s-1}+ \sum_{j=\max(1,s-H)}^{s-1} \AK^{j-1}B M^{[s-j]}(K).
\end{align*}
Specifically, when $s\leq H$, we have
\begin{align*}
\tilde \Phi^x_{s}(\bm M(K))& =\AK^{s-1}+ \sum_{j=1}^{s-1}\AK^{j-1}B M^{[s-j]}(K)\\
& = \AK^{s-1}+ \sum_{j=1}^{s-1}\AK^{j-1}B (\Kb-K)A_{K}^{s-j-1}\\
&= \AK^{s-1}+ \sum_{j=1}^{s-1}\AK^{j-1}(A_{K}-\AK)A_{K}^{s-j-1}\\
&=\AK^{s-1}+ \sum_{j=1}^{s-1}\AK^{j-1}A_{K}^{s-j}-\sum_{j=1}^{s-1}\AK^{j}A_{K}^{s-j-1}\\
&=\AK^{s-1}+A_{K}^{s-1}-\AK^{s-1}=A_{K}^{s-1}.
\end{align*}
When $s>H$, 
\begin{align*}
\tilde \Phi^x_{s}(\bm M(K))& =\AK^{s-1}+ \sum_{j=s-H}^{s-1}\AK^{j-1}B M^{[s-j]}(K)\\
&=\AK^{s-1}+ \sum_{j=s-H}^{s-1}\AK^{j-1}B (\Kb-K)A_{K}^{s-j-1}\\
&=\AK^{s-1}+ \sum_{j=s-H}^{s-1}\AK^{j-1}A_{K}^{s-j}-\sum_{j=s-H}^{s-1}\AK^{j}A_{K}^{s-j-1}\\
&=\AK^{s-H-1}A_{K}^H
\end{align*}
Therefore,
\begin{align*}
x_t^{\bm M(K)}=\sum_{s=1}^H  A_{K}^{s-1} w_{t-s}+ \sum_{s=H+1}^t \AK^{s-H-1}A_{K}^H w_{t-s}
\end{align*}
Accordingly, we have a formula for $u_t^{\bm M(K)}$ when implementing $\bm M(K)$.
\begin{align*}
u_t^{\bm M(K)}&=-\Kb x_t^{\bm M(K)} + \sum_{s=1}^H M^{[s]}(K)w_{t-s}\\
& = -\sum_{s=1}^H \Kb A_{K}^{s-1} w_{t-s}- \sum_{s=H+1}^t \Kb \AK^{s-H-1}A_{K}^H w_{t-s}+\sum_{s=1}^H (\Kb-K)A_{K}^{s-1}w_{t-s}\\ 
&=-\sum_{s=1}^H K A_{K}^{s-1} w_{t-s}- \sum_{s=H+1}^t \Kb \AK^{s-H-1}A_{K}^H w_{t-s}
\end{align*}

Therefore, 
\begin{align*}
\|x_t^K-x_t^{\bm M(K)}\|_2&= \|\sum_{s=H+1}^t (\AK^{s-H-1}-A_{K}^{s-H-1})A_{K}^H w_{t-s}\|_2\\
& \leq \sum_{s=H+1}^t (\|\AK^{s-H-1}\|+\|A_{K}^{s-H-1}\|_2 )\|A_{K}^H\|_2 \sqrt n \bar w\\
& \leq \sum_{s=H+1}^t 2 \kappa^2(1-\gamma)^{s-H-1} \kappa^2 (1-\gamma)^H \sqrt n \bar w\\
& \leq 2\kappa^4 \sqrt n \bar w(1-\gamma)^H/\gamma
\end{align*}
and 
\begin{align*}
\|u_t^K-u_t^{\bm M(K)}\|_2&=\|\sum_{s=H+1}^t (\Kb\AK^{s-H-1}-KA_{K}^{s-H-1})A_{K}^H w_{t-s}\|_2\\
& \leq \sum_{s=H+1}^t 2 \kappa^5 (1-\gamma)^{s-H-1}\sqrt n \bar w (1-\gamma)^H\\
& \leq 2\kappa^5 \sqrt n \bar w (1-\gamma)^H/\gamma
\end{align*}
Next, we verify that $\bm M(K) \in \M$:
\begin{align*}
\|M^{[i]}(K)\|_\infty& \leq \sqrt n 	\|M^{[i]}(K)\|_2 \leq 2\sqrt n \kappa \kappa^2(1-\gamma)^{i-1}=2\sqrt n \kappa^3(1-\gamma)^{i-1}.
\end{align*}

Finally, we prove the loose feasibility. 
Notice that 
\begin{align*}
\|D_x x_t^{\bm M(K)}-D_x x_t^K\|_\infty &\leq \|D_x \|_\infty \|x_t^{\bm M(K)} -x_t^K\|_\infty \\
& \leq  \|D_x \|_\infty \|x_t^{\bm M(K)} -x_t^K\|_2 \\
& \leq \|D_x\|_\infty 2\kappa^4 \sqrt n \bar w(1-\gamma)^H/\gamma
\end{align*}
Therefore, for any $1\leq i \leq k_x$, we have
\begin{align*}
D_{x,i}^\top x_t^{\bm M(K)}\leq \|D_x x_t^{\bm M(K)}-x_t^K\|_\infty + D_{x,i}^\top x_t^K \leq \|D_x\|_\infty 2\kappa^4 \sqrt n \bar w(1-\gamma)^H/\gamma+ d_{x,i}
\end{align*}
Similarly, 
\begin{align*}
D_{u,j}^\top u_t^{\bm M(K)}\leq \|D_u u_t^{\bm M(K)}-D_u u_t^K\|_\infty + D_{u,j}^\top u_t^K \leq \|D_u\|_\infty 2\kappa^5 \sqrt n \bar w(1-\gamma)^H/\gamma+ d_{u,j}
\end{align*}
Let $\epsilon_3(H)=\max(\|D_x\|_\infty, \|D_u\|_\infty)2\kappa^5 \sqrt n \bar w(1-\gamma)^H/\gamma$. This completes the proof.

\subsection{Proof of Corollary \ref{cor: K's M in Omega}}
	Let $\{(x_t^K, u_t^K)\}_{t=0}^T $ and $\{(x_t^{\bm M(K)}, u_t^{\bm M(K)})\}_{t=0}^T $ denote the state-action trajectory under linear controller $K$ and disturbance-action policy $\bm M$ respectively.  Similar to Lemma \ref{lem: define M_ap and epsilon3}, we have $D_x x_t^{\bm M(K)} \leq d_x-\epsilon_0\one_{k_x}+\epsilon_3\one_{k_x}$ and $D_u u_t^{\bm M(K)} \leq d_u-\epsilon_0\one_{k_u}+\epsilon_3\one_{k_u}$ for any disturbances $\{w_k \in \W\}_{k=0}^T$. Since $D_x x_t^{\bm M(K)}=D_x \AK^H x^{\bm M(K)}_{t-H}+D_x \tilde x_t^{\bm M(K)}$, by Lemma \ref{lem: bdd xt-H part}, we have  $D_x \tilde x_{t}^{\bm M(K)}\leq d_x +(\epsilon_1+\epsilon_3-\epsilon_0)\one_{k_x}$. Similarly, we can show that $D_u \tilde u_{t}^{\bm M(K)}\leq d_u +(\epsilon_1+\epsilon_3-\epsilon_0)\one_{k_u}$. Following the procedures in Step 2 of Section \ref{sec: algorithm} and noticing that $\bm M(K)$ is time-invariant, we can show $	\bm M(K)\in \Omega_{\epsilon_0-\epsilon_1-\epsilon_3}$ by the definitions. 
\section{Proofs of lemmas used in the proof of Theorem \ref{thm: regret bdd}}

\subsection{Proof of Lemma \ref{lem:  part i}}
We first prove the bound on Part i (Lemma \ref{lem: part i}) because we will use this bound when proving Lemma \ref{lem: part iii second half}.

To bound Part i, we consider the following two lemmas.
\begin{lemma}\label{lem: c and f's difference}	
	For $\bm M_t \in \M$ for all $t$, we have
	\begin{align*}\label{lem: c and f's difference}
	\left| J_T(\bm M_{0:T})- \sum_{t=0}^T f_t(\bm M_{t-H:t})\right| \leq T G b^2 \kappa^2 (1-\gamma)^H (1+\kappa)
	\end{align*}
\end{lemma}
\begin{proof}
	Remember that $f_t(\bm M_{t-H:t})=\E[c_t(\tilde x_t, \tilde u_t)]$. By Proposition \ref{prop: approx}, Lemma \ref{lem: bdd on xt ut xt tilde ut tilde}, and Assumption \ref{ass: bounded Hessian largest evalue}, when $t\geq 1$, we have
	\begin{align*}
	| \E[c_t(x_t,u_t)]-\E[c_t(\tilde x_t, \tilde u_t)]| & \leq \E[|c_t(x_t,u_t)-c_t(\tilde x_t, \tilde u_t)|]\\
	& \leq \E [\left| \langle \nabla_x c_t(x_t, u_t), x_t-\tilde x_t \rangle \right|]+  \E [\left| \langle \nabla_u c_t(x_t, u_t), u_t-\tilde u_t \rangle \right|]\\
	& \leq \E[Gb(\|x_t-\tilde x_t\|_2+\|u_t-\tilde u_t\|_2)]\\
	& \leq Gb\E[\|\AK^H x_{t-H}\|_2+ \|\Kb \AK^H x_{t-H}\|_2]\\
	& \leq Gb (\kappa^2+\kappa^3) (1-\gamma)^H b 
	\end{align*}
	When $t=0$, since $c_0(0,0)=0$, $x_0=\tilde x_0=0$, and $u_0=\tilde u_0=0$, we have $\E[c_0(x_0, u_0)]-\E[c_0(\tilde x_0, \tilde u_0)]=0$. 
	The proof is completed by summing over $t=0, \dots, T$.
\end{proof}
\begin{lemma}
	Apply Algorithm \ref{alg:ogd} with constant stepsize $\eta$. Then, 
	\begin{align*}
	\left| \sum_{t=0}^T f_t(\bm M_{t-H:t})-\sum_{t=0}^T \mathring f_t(M_t)\right| \leq Gb(1+\kappa)\sqrt n \bar w\kappa^2\kappa_B \sqrt H G_f \eta \frac{1}{\gamma^2}T
	\end{align*}
\end{lemma}
\begin{proof}
	Firstly, by OGD's procedures, we have $\|\bm M_t -\bm M_{t-k}\|_F\leq k \eta G_f$. 
	
	Next, by Lemma \ref{lem: ft(Mt) property}, and since $\mathring f_t(\bm M)=f_t(\bm M, \dots, \bm M)$, we have
	\begin{align*}
	| f_t(\bm M_{t-H:t})-\mathring f_t(\bm M_t)| &\leq Gb(1+\kappa)\sqrt n \bar w\kappa^2\kappa_B \sqrt H \sum_{i=0}^H(1-\gamma)^{\max(i-1,0)}  \|\bm M_{t-i}-{\bm M}_{t}\|_F\\
	& \leq Gb(1+\kappa)\sqrt n \bar w\kappa^2\kappa_B \sqrt H G_f \eta \frac{1}{\gamma^2}
	\end{align*}
	when $1\leq t \leq T$ and $f_t(\bm M_{t-H:t})-\mathring f_t(\bm M_t)=0$ when $t=0$. Therefore, by summing over $1\leq t \leq T$, we complete the proof.
\end{proof}
In conclusion, we can bound Part i by summing up the  bounds in the two lemmas above.

\subsection{Proof of Lemma \ref{lem:  part ii}}
OGD's regret bound is standard in the literature (see e.g. \cite{hazan2019introduction}). The bound on $G_f$ is proved in Lemma \ref{lem: ft(Mt) property}. Next, we bound the diameter of $\Omega_\epsilon$. 	For any $\bm M, \tilde{\bm M} \in \Omega_\epsilon$, we have $\bm M, \tilde{\bm M} \in \M$. Therefore, 
\begin{align*}
\|\bm M-\tilde{\bm M}\|_F& \leq \sum_{i=1}^H \|M^{[i]}-\tilde M^{[i]}\|_F \leq \sum_{i=1}^H (\|M^{[i]}\|_F+\|\tilde M^{[i]}\|_F)\leq \sqrt m \sum_{i=1}^H (\|M^{[i]}\|_\infty+\|\tilde M^{[i]}\|_\infty)\\
& \leq \sqrt m \sum_{i=1}^H 4 \sqrt n \kappa^3 (1-\gamma)^{i-1} =4 \sqrt{mn} \kappa^3/\gamma=\delta
\end{align*}
\subsection{Proof of Lemma \ref{lem: part iii second half}}
For notational simplicity, we slightly abuse the notation and let $(x_t, u_t)$ denote the trajectory generated by $\bm M_{\text{ap}}$ in this proof. 

By Lemma \ref{lem: bdd on xt ut xt tilde ut tilde}, the proof of Lemma \ref{lem: define M_ap and epsilon3}, and Assumption \ref{ass: bounded Hessian largest evalue}, we can bound $	J_T(\bm M_{\text{ap}})-J_T(K^*)$.
\begin{align*}
J_T(\bm M_{\text{ap}})-J_T(K^*)&=\sum_{t=0}^T\E[c_t(x_t,u_t)-c_t(x_t^*,u_t^*)]\\
& =\sum_{t=1}^T\E[c_t(x_t,u_t)-c_t(x_t^*,u_t^*)]\\
& \leq \sum_{t=1}^T G b 2\kappa^4(1+\kappa) \sqrt n \bar w (1-\gamma)^H/\gamma\\
&=2TG b \kappa^4(1+\kappa) \sqrt n \bar w (1-\gamma)^H/\gamma
\end{align*}
Further, by Lemma \ref{lem: c and f's difference}, we have
\begin{align*}
\sum_{t=0}^T \mathring f_t(\bm M_{\text{ap}})-J_T(K^*) &= \sum_{t=0}^T \mathring f_t(\bm M_{\text{ap}})-J_T(\bm M_{\text{ap}})+J_T(\bm M_{\text{ap}})-J_T(K^*) 
\\
& \leq T G b^2 \kappa^2 (1-\gamma)^H (1+\kappa)+2TG b \kappa^4(1+\kappa) \sqrt n \bar w (1-\gamma)^H/\gamma\\
& \leq \Theta(T n^2 m H^2(1-\gamma)^H)
\end{align*}
when $H\geq \frac{\log(2\kappa^2)}{\log((1-\gamma)^{-1})}$.